\documentclass[11pt]{article}
\usepackage[margin=1in]{geometry}
\usepackage[colorlinks=true, allcolors=blue]{hyperref}

\usepackage{fixmath}
\usepackage{bm}
\usepackage{amsbsy}
\usepackage{color}
\usepackage{verbatim}
\usepackage{multirow}
\usepackage{amssymb}
\usepackage{amsthm}
\usepackage{array}
\usepackage{mathtools}
\usepackage{amsmath}
\usepackage{graphicx}
\usepackage{tikz}
\usetikzlibrary{positioning}
\usepackage{xcolor}

\usepackage{thm-restate}

\newtheorem{theorem}{Theorem}
\newtheorem{lemma}[theorem]{Lemma}

\newcommand{\size}[1]{\ensuremath{|#1|}}

\newcommand{\lrA}[1]{\ensuremath{\left(#1\right)}}
\newcommand{\lrB}[1]{\ensuremath{\left[#1\right]}}

\def\LB{\mbox{LB}}

\newcommand{\EE}[1]{\ensuremath{\mathbb{E}[#1]}}
\newcommand{\EEE}[1]{\ensuremath{\mathbb{E}\lrB{#1}}}

\title{Practical Algorithms with Guaranteed Approximation Ratio for TTP with Maximum Tour Length Two}

\author
{
Jingyang Zhao\footnote{University of Electronic Science and Technology of China. Email: \texttt{jingyangzhao1020@gmail.com}.}
\and
Mingyu Xiao\footnote{University of Electronic Science and Technology of China. Email: \texttt{myxiao@gmail.com}.}
}

\date{}

\begin{document}
\maketitle

\begin{abstract}
The Traveling Tournament Problem (TTP) is a hard but interesting sports scheduling problem inspired by Major League Baseball, which is to design a double round-robin schedule such that each pair of teams plays one game in each other's home venue, minimizing the total distance traveled by all $n$ teams ($n$ is even). In this paper, we consider TTP-2, i.e., TTP under the constraint that at most two consecutive home games or away games are allowed for each team. We propose practical algorithms for TTP-2 with improved approximation ratios. Due to the different structural properties of the problem, all known algorithms for TTP-2 are different for $n/2$ being odd and even, and our algorithms are also different for these two cases. For even $n/2$, our approximation ratio is $1+3/n$, improving the previous result of $1+4/n$. For odd $n/2$, our approximation ratio is $1+5/n$, improving the previous result of $3/2+6/n$. In practice, our algorithms are easy to implement. Experiments on well-known benchmark sets show that our algorithms beat previously known solutions for all instances with an average improvement of $5.66\%$.
\end{abstract}

\maketitle

\section{Introduction}
The Traveling Tournament Problem (TTP), first systematically introduced in~\cite{easton2001traveling}, is a hard but interesting sports scheduling problem inspired by Major League Baseball.
This problem is to find a double round-robin tournament satisfying several constraints that minimizes the total distances traveled by all participant teams.
There are $n$ participating teams in the tournament, where $n$ is always even. Each team should play $2(n-1)$ games in $2(n-1)$ consecutive days. Since each team can only play one game on each day, there are exact $n/2$ games scheduled on each day.
There are exact two games between any pair of teams,
where one game is held at the home venue of one team and the other one is held at the home venue of the other team.
The two games between the same pair of teams could not be scheduled in two consecutive days.
These are the constraints for TTP. We can see that it is not easy to construct a feasible schedule.
Now we need to find an optimal schedule that minimizes the total traveling distances by all the $n$ teams.
A well-known variant of TTP is TTP-$k$, which has one more constraint:
each team is allowed to take at most $k$ consecutive home or away games.
If $k$ is very large, say $k=n-1$, then this constraint will lose its meaning and it becomes TTP again. For this case, a team can schedule its travel distance as short as the solution to the traveling salesmen problem. On the other hand,
in a sports schedule, it is generally believed that home stands and road trips should alternate as regularly as possible for each team~\cite{campbell1976minimum,thielen2012approximation}.
The smaller the value of $k$, the more frequently teams have to return their homes.
TTP and its variants have been extensively studied in the literature~\cite{kendall2010scheduling,rasmussen2008round,thielen2012approximation,DBLP:conf/mfcs/XiaoK16}.

\subsection{Related Work}
In this paper, we will focus on TTP-2. We first survey the results on TTP-$k$.
For $k=1$, TTP-1 is trivial and there is no feasible schedule~\cite{de1988some}.
But when $k\geq 2$, the problem becomes hard. Feasible schedules will exist but it is not easy to construct one. Even no good
brute force algorithm with single exponential running time has been found yet.
In the online benchmark~\cite{trick2007challenge} (there is also a new benchmark website, updated by Van Bulck \emph{et al.}~\cite{DBLP:journals/eor/BulckGSG20}), most instances with more than $10$ teams are still unsolved completely even by using high-performance machines.
The NP-hardness of TTP was proved in \cite{bhattacharyya2016complexity}.
TTP-3 was also proved to be NP-hard \cite{thielen2011complexity} and the idea of the proof can be extended to prove the NP-hardness of TTP-$k$ for each constant $k\geq 4$~\cite{DBLP:journals/corr/abs-2110-02300}.
Although the hardness of TTP-2 has not been theoretically proved yet, most people believe TTP-2 is also hard since no single exponential algorithm to find an optimal solution to TTP-2 has been found after 20 years of study.
In the literature, there is a large number of contributions on approximation algorithms~\cite{miyashiro2012approximation,yamaguchi2009improved,imahori2010approximation,westphal2014,hoshino2013approximation,thielen2012approximation,DBLP:conf/mfcs/XiaoK16} and heuristic algorithms~\cite{easton2003solving,lim2006simulated,anagnostopoulos2006simulated,di2007composite,goerigk2014solving}.

In terms of approximation algorithms, most results are based on the assumption that the distance holds the symmetry and triangle inequality properties. This is natural and practical in the sports schedule.
For TTP-3, the first approximation algorithm, proposed by Miyashiro \emph{et al.}, admits a $2+O(1/n)$ approximation ratio~\cite{miyashiro2012approximation}.
They first proposed a randomized $(2+O(1/n))$-approximation algorithm and then derandomized the algorithm without changing the approximation ratio~\cite{miyashiro2012approximation}.
Then, the approximation ratio was improved to $5/3+O(1/n)$ by Yamaguchi \emph{et al.}~\cite{yamaguchi2009improved} and to $139/87+O(1/n)$ by Zhao \emph{et al.}~\cite{zhao2022improved}.
For TTP-4, the approximation ratio has been improved to $17/10+O(1/n)$~\cite{zhao2022improved}. For TTP-$k$ with $k\geq 5$, the approximation ratio has been improved to $(5k-7)/(2k)+O(k/n)$~\cite{imahori2010approximation}.
For TTP-$k$ with $k\geq n-1$, Imahori\emph{ et al.} \cite{imahori2010approximation} proved an approximation ratio of 2.75. At the same time, Westphal and Noparlik \cite{westphal2014} proved an approximation ratio of 5.875 for any choice of $k\geq 4$ and $n\geq 6$.

In this paper, we will focus on TTP-2.
The first record of TTP-2 seems from the schedule of a basketball conference of ten teams
in~\cite{campbell1976minimum}. That paper did not discuss the approximation ratio.
In fact, any feasible schedule for TTP-2
is a 2-approximation solution under the metric distance~\cite{thielen2012approximation}.
Although any feasible schedule will not have a very bad performance, no simple construction of feasible schedules is known now.
In the literature, all known algorithms for TTP-2 are different for $n/2$ being even and odd. This may be caused by different structural properties. One significant contribution to TTP-2 was done by Thielen and Westphal~\cite{thielen2012approximation}.
They proposed a $(1+16/n)$-approximation algorithm for $n/2$ being even and a $(3/2+6/n)$-approximation algorithm for $n/2$ being odd, and asked as an open problem whether the approximation ratio could be improved to $1+O(1/n)$ for the case that $n/2$ is odd. For even $n/2$, the approximation ratio was improved to $1+4/n$ by Xiao and Kou~\cite{DBLP:conf/mfcs/XiaoK16}.
There is also a known algorithm with the approximation ratio $1+\frac{\lceil\log_2 {n}\rceil+2}{2(n-2)}$, which is better for $n\leq 32$~\cite{DBLP:conf/atmos/ChatterjeeR21}.
For odd $n/2$, two papers solved Thielen and Westphal's open problem independently by giving algorithms with approximation ratio $1+O(1/n)$: Imahori~\cite{imahori20211+} proposed an idea to solve the case of $n\geq 30$ with an approximation ratio of $1+24/n$; in a preliminary version of this paper~\cite{DBLP:conf/ijcai/ZhaoX21}, we provided a practical algorithm with an approximation ratio of $1+12/n$. In this version, we will further improve the approximation ratio by using refined analysis.

\subsection{Our Results}
In this paper, we design two practical algorithms for TTP-2, one for even $n/2$ and one for odd $n/2$.

For even $n/2$, we first propose an algorithm with an approximation ratio of $1+\frac{3}{n}-\frac{10}{n(n-2)}$ by using the packing-and-combining method.
Then, we apply a divide-and-conquer method to our packing-and-combining method, and propose a more general algorithm with an approximation ratio $1+\frac{3}{n}-\frac{18}{n(n-2)}$ for $n\geq 16$ and $n\equiv0 \pmod 8$. Our results improve the previous result of $1+\frac{4}{n}+\frac{4}{n(n-2)}$ in~\cite{DBLP:conf/mfcs/XiaoK16}.
For odd $n/2$, we prove an approximation ratio of $1+\frac{5}{n}-\frac{10}{n(n-2)}$, improving the result of $\frac{3}{2}+\frac{6}{n-4}$ in~\cite{thielen2012approximation}.
In practice, our algorithms are easy to implement and run very fast.
Experiments show that our results can beat all previously-known solutions on the 33 tested benchmark instances in \cite{trick2007challenge}: for even $n/2$ instances, the average improvement is $2.86\%$; for odd $n/2$ instances, the average improvement is $8.65\%$.

Partial results of this paper were presented at the 27th International Computing and Combinatorics Conference (COCOON 2021)~\cite{DBLP:conf/cocoon/ZhaoX21} and the 30th International Joint Conference on Artificial Intelligence (IJCAI 2021)~\cite{DBLP:conf/ijcai/ZhaoX21}. In \cite{DBLP:conf/cocoon/ZhaoX21}, we proved an approximation ratio of $1+\frac{3}{n}-\frac{6}{n(n-2)}$ for TTP-2 with even $n/2$. In \cite{DBLP:conf/ijcai/ZhaoX21}, we proved an approximation ratio of $1+\frac{12}{n}+\frac{8}{n(n-2)}$ for TTP-2 with odd $n/2$. In this paper, we make further improvements for both cases. In the experiments, we also get a better performance.

\section{Preliminaries}\label{sec_pre}
We will always use $n$ to denote the number of teams and let $m=n/2$, where $n$ is an even number.
We use $\{t_1, t_2, \dots, t_n\}$ to denote the set of the $n$ teams.
A sports scheduling on $n$ teams is \emph{feasible} if it holds the following properties.
\begin{itemize}
\item \emph{Fixed-game-value}: Each team plays two games with each of the other $n-1$ teams, one at its home venue and one at its opponent's home venue.
\item \emph{Fixed-game-time}: All the games are scheduled in $2(n-1)$ consecutive days and each team plays exactly one game in each of the $2(n-1)$ days.
\item \emph{Direct-traveling}: All teams are initially at home before any game begins, all teams will come back home after all games, and a team travels directly from its game venue in the $i$th day to its game venue in the $(i+1)$th day.
\item \emph{No-repeat}: No two teams play against each other on two consecutive days.
\item \emph{Bounded-by-$k$}: The number of consecutive home/away games for any team is at most $k$.
\end{itemize}

The TTP-$k$ problem is to find a feasible schedule minimizing the total traveling distance of all the $n$ teams.
The input of TTP-$k$ contains an $n \times n$ distance matrix $D$ that indicates the distance between each pair of teams.
The distance from the home of team $i$ to the home of team $j$ is denoted by $D_{i,j}$.
We also assume that $D$ satisfies the symmetry and triangle inequality properties, i.e., $D_{i,j}=D_{j,i}$ and $D_{i,j} \leq D_{i,h} + D_{h,j}$ for all $i,j,h$. We also let $D_{i,i}=0$ for each $i$.

We will use $G$ to denote an edge-weighted complete graph on $n$ vertices representing the $n$ teams.
The weight of the edge between two vertices $t_i$ and $t_j$ is $D_{i,j}$, the distance from the home of $t_i$ to the home of $t_j$.
We also use $D_i$ to denote the weight sum of all edges incident on $t_i$ in $G$, i.e., $D_i=\sum_{j=1}^n D_{i,j}$.
The sum of all edge weights of $G$ is denoted by $D_G$.

We let $M$ denote a minimum weight perfect matching in $G$. The weight sum of all edges in $M$ is denoted by $D_M$.
We will consider the endpoint pair of each edge in $M$ as a \emph{super-team}. We use $H$ to denote the complete graph on the $m$ vertices representing the $m$ super-teams. The weight of the edge between two super-teams $u_i$ and $u_j$, denoted by $D(u_i,u_j)$, is the sum of the weight of the four edges in $G$ between one team in $u_i$ and one team in $u_j$, i.e., $D(u_i, u_j)=\sum_{t_{i'}\in u_i \& t_{j'}\in u_j}D_{i',j'}$. We also let $D(u_i,u_i)=0$ for any $i$.
We give an illustration of graphs $G$ and $H$ in Figure~\ref{fig001}.
\begin{figure}[ht]
    \centering
    \includegraphics[scale=0.8]{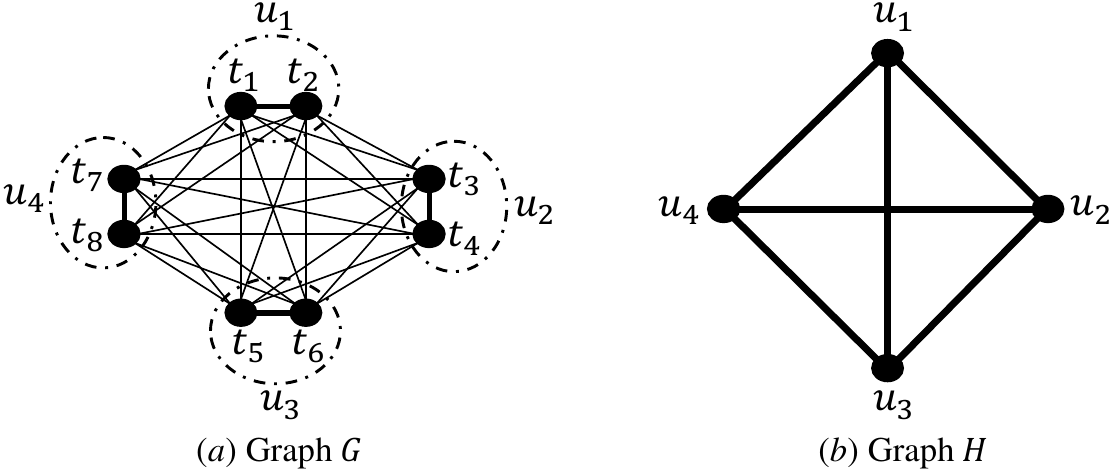}
    \caption{An illustration of graphs $G$ and $H$, where there four dark lines form a minimum perfect matching $M$ in $G$}
    \label{fig001}
\end{figure}

The sum of all edge weights of $H$ is denoted by $D_H$. It holds that
\begin{eqnarray} \label{eqn_GH}
D_H=D_G-D_M.
\end{eqnarray}

\subsection{Independent Lower Bound and Extra Cost}
The \emph{independent lower bound} for TTP-2 was firstly introduced by Campbell and Chen~\cite{campbell1976minimum}.
It has become a frequently used lower bound.
The basic idea of the independent lower bound is to obtain a lower bound $LB_i$ on the traveling distance of a single team $t_i$ independently without considering the feasibility of other teams.

The road of a team $t_i$ in TTP-$2$, starting at its home venue and coming back home after all games, is called
an \emph{itinerary} of the team. The itinerary of $t_i$ is also regarded as a graph on the $n$ teams,
which is called the \emph{itinerary graph} of $t_i$.
In an itinerary graph of $t_i$, the degree of all vertices except $t_i$ is 2 and the degree of $t_i$ is greater than or equal to $n$ since team $t_i$ will visit each other team venue only once.
Furthermore, for any other team $t_j$, there is at least one edge between $t_i$ and $t_j$, because $t_i$ can only visit at most 2 teams on each road trip and then team $t_i$ either comes from its home to team $t_j$ or goes back to its home after visiting team $t_j$. We decompose the itinerary graph of $t_i$ into two parts: one is a spanning star centered at $t_i$ (a spanning tree which only vertex $t_i$ of degree $> 1$) and the forest of the remaining part. Note that in the forest, only $t_i$ may be a vertex of degree $\geq 2$ and all other vertices are degree-1 vertices. See Figure~\ref{fig002} for illustrations of the itinerary graphs.

\begin{figure}[ht]
    \centering
    \includegraphics[scale=1.1]{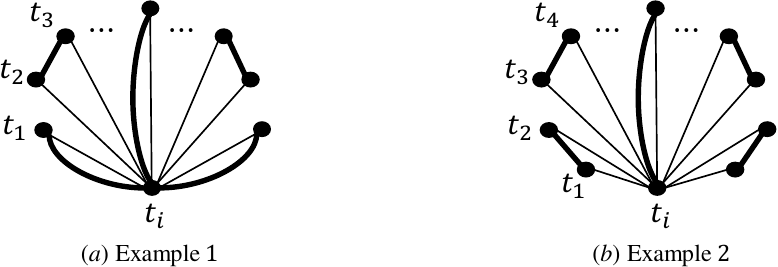}
    \caption{The itinerary graph of $t_i$, where the light edges form a spanning star and the dark edges form the remaining forest. In the right example (b), the remaining forest is a perfect matching of $G$}
    \label{fig002}
 \end{figure}

For different itineraries of $t_i$, the spanning star is fixed and only the remaining forest may be different.
The total distance of the spanning star is $\sum_{j\neq i} D_{i,j}=D_i$. Next, we show an upper and lower bound on the total distance of the remaining forest. For each edge between two vertices $t_{j_1}$ and $t_{j_2}$ ($j_1,j_2\neq i$), we have that $D_{j_1,j_2}\leq D_{i,j_1}+D_{i,j_2}$ by the triangle inequality property. Thus, we know that the total distance of the remaining forest is at most the total distance of the spanning star. Therefore, the distance of any feasible itinerary of $t_i$ is at most $2D_i$.

\begin{lemma}\label{perfecti}
The traveling distance of any itinerary of a team $t_i$ is at most $2D_i$.
\end{lemma}

Lemma~\ref{perfecti} implies that the worst itinerary only consists of road trips containing one game.

On the other hand, the distance of the remaining forest is at least as that of a minimum perfect matching of $G$ by the triangle inequality.
Recall that we use $M$ to denote a minimum perfect matching of $G$.

\begin{lemma}\label{perfecti+}
The traveling distance of any itinerary of a team $t_i$ is at least $D_i+D_M$.
\end{lemma}

Thus, we have a lower bound $LB_i$ for each team $t_i$:
\begin{eqnarray} \label{eqn_lower1}
LB_i=D_i+D_M.
\end{eqnarray}

The itinerary of $t_i$ to achieve $LB_i$ is called the \emph{optimal itinerary}.
The \emph{independent lower bound} for TTP-2 is the traveling distance such that all teams reach their optimal itineraries, which is denoted as
\begin{eqnarray} \label{eqn_lowerbound}
 LB=\sum_{i=1}^n LB_i =\sum_{i=1}^n (D_i +D_M)=2D_G+nD_M.
\end{eqnarray}

Lemma~\ref{perfecti} and (\ref{eqn_lower1}) imply that
\begin{lemma}\label{useful}
The traveling distance of any feasible itinerary of a team $t_i$ is at most $2D_i$.
\end{lemma}

Hence, we can get that
\begin{theorem} \label{twoapp}
Any feasible schedule for TTP-2 is a 2-approximation solution.
\end{theorem}

Theorem~\ref{twoapp} was first proved in~\cite{thielen2012approximation}.
For any team, it is possible to reach its optimal itinerary. However, it is impossible for all teams to reach their optimal itineraries synchronously in a feasible schedule even for $n=4$~\cite{thielen2012approximation}. It is easy to construct an example. So the independent lower bound for TTP-2 is not achievable.

To analyze the quality of a schedule of the tournament, we will compare the itinerary of each team with the optimal itinerary.
The different distance is called the \emph{extra cost}. Sometimes it is not convenient to compare the whole itinerary directly.
We may consider the extra cost for a subpart of the itinerary.
We may split an itinerary into several trips and each time we compare some trips.
A \emph{road trip} in an itinerary of team $t_i$ is a simple cycle starting and ending at $t_i$.
So an itinerary consists of several road trips. For TTP-2, each road trip is a triangle or a cycle on two vertices.
Let $L$ and $L'$ be two itineraries of team $t_i$, $L_s$ be a sub itinerary of $L$ consisting of several road trips in $L$, and
$L'_s$ be a sub itinerary of $L'$ consisting of several road trips in $L'$.
We say that the sub itineraries $L_s$ and $L'_s$ are \emph{coincident} if they visit the same set of teams.
We will only compare a sub itinerary of our schedule with a coincident sub itinerary of the optimal itinerary and consider the extra cost between them.

\section{Framework of The Algorithms}
Our algorithms for even $n/2$ and odd $n/2$ are different due to the different structural properties.
However, the two algorithms have a similar framework. We first introduce the common structure of our algorithms.

\subsection{The Construction}
The cases that $n=4$ and $6$ can be solved easily by a brute force method. For the sake of presentation, we assume that the number of teams $n$ is at least $8$.

Our construction of the schedule for each case consists of two parts. First we arrange \emph{super-games} between \emph{super-teams}, where each super-team contains
a pair of normal teams. Then we extend super-games to normal games between normal teams.
To make the itinerary as similar as the optimal itinerary, we may take each pair of teams in the minimum perfect matching $M$ of $G$ as a \emph{super-team}. There are $n$ normal teams and then there are $m=n/2$ super-teams. Recall that we use $\{u_1, u_2, \dots, u_{m-1}, u_{m}\}$ to denote the set of super-teams and relabel the $n$ teams such that $u_i=\{t_{2i-1},t_{2i}\}$ for each $i$.

Each super-team will attend $m-1$ super-games in $m-1$ time slots.
Each super-game in the first $m-2$ time slots will be extended to normal games between normal teams on four days, and each super-game in the last time slot will be extended to normal games between normal teams on six days. So each normal team $t_i$ will attend $4\times (m-2)+6=4m-2=2n-2$ games.
This is the number of games each team $t_i$ should attend in TTP-2.

For even $n/2$ and odd $n/2$, the super-games and the way to extend super-games to normal games will be different.

\subsection{The Order of Teams}
To get a schedule with a small total traveling distance, we order teams such that the pair of teams in each super-team corresponds to an edge in the minimum matching $M$. However, there are still $m!\cdot 2^m$ choices to order all the $n$ teams, where there are $m!$ choices to order super-teams $\{u_1,\dots,u_m\}$ and $2^m$ choices to order all teams in $m$ super-teams (there are two choices to order the two teams in each super-team). To find an appropriate order, we propose a simple randomized algorithm which contains the following four steps.

\medskip

\noindent\textbf{Step~1.} Compute a minimum perfect matching $M$ of $G$.

\noindent\textbf{Step~2.} Randomly sort all $m$ edges in $M$ and get a set of super-teams $\{u_1,\dots,u_m\}$ by taking the pair of teams corresponding to the $i$-th edge of $M$ as super-team $u_i$.

\noindent\textbf{Step~3.} Randomly order teams $\{t_{2i-1}, t_{2i}\}$ in each super-team $u_i$.

\noindent\textbf{Step~4.} Apply the order of $n$ teams to our construction.

\medskip

The randomized versions of the algorithms are easier to present and analyze. We show that the algorithms can be derandomized efficiently by using the method of conditional expectations~\cite{motwani1995randomized}.

\begin{lemma}\label{core}
Assuming $t_i\in u_{i'}$ and $t_j\in u_{j'}$ where $i'\neq j'$, it holds that $\EE{D_{i,j}}\leq\frac{1}{n(n-2)}\LB$.
\end{lemma}
\begin{proof}
An edge $t_it_j$ with $i'\neq j'$ corresponds to an edge in $G-M$.
Since we label super-teams and teams in each super-team randomly, the probability of each edge in $G-M$ being $t_it_j$ is $\frac{1}{\size{G-M}}$.
Hence, the expected weight of the edge $t_it_j$ is
$$\frac{1}{\size{G-M}}(D_G-D_M)=\frac{1}{n(n-1)/2-n/2}D_H=\frac{2}{n(n-2)}D_H\leq \frac{1}{n(n-2)}\LB,$$
because $\LB=2D_G+nD_M\geq 2D_G\geq 2D_H$ by (\ref{eqn_GH}) and (\ref{eqn_lowerbound}).
\end{proof}

Next, we will show the details of our constructions.

\section{The Construction for Even $n/2$}
In this section, we study the case of even $n/2$. We will first introduce the construction of the schedule. Then, we analyze the approximation quality of the randomized algorithm. At last, we will propose a dive-and-conquer method to get some further improvements.

\subsection{Construction of the Schedule}
We construct the schedule for super-teams from the first time slot to the last time slot $m-1$.
In each of the $m-1$ time slots, there are $\frac{m}{2}$ super-games.
In total, we have four different kinds of super-games: \emph{normal super-games}, \emph{left super-games}, \emph{penultimate super-games}, and \emph{last super-games}.
Each of the first three kinds of super-games will be extended to eight normal games on four consecutive days. Each last super-game will be extended to twelve normal games on six consecutive days.
We will indicate what kind each super-game belongs to.

For the first time slot, the $\frac{m}{2}$ super-games are arranged as shown in Figure~\ref{figa1}. Super-team $u_i$ plays against super-team $u_{m-1-i}$ for $i\in \{1, \dots, m/2-1\}$ and super-team $u_{m-1}$ plays against super-team $u_{m}$. Each super-game is represented by a directed edge, the direction information of which will be used to extend super-games to normal games between normal teams.
All the super-games in the first time slot are normal super-games.

\begin{figure}[ht]
    \centering
    \includegraphics[scale=0.55]{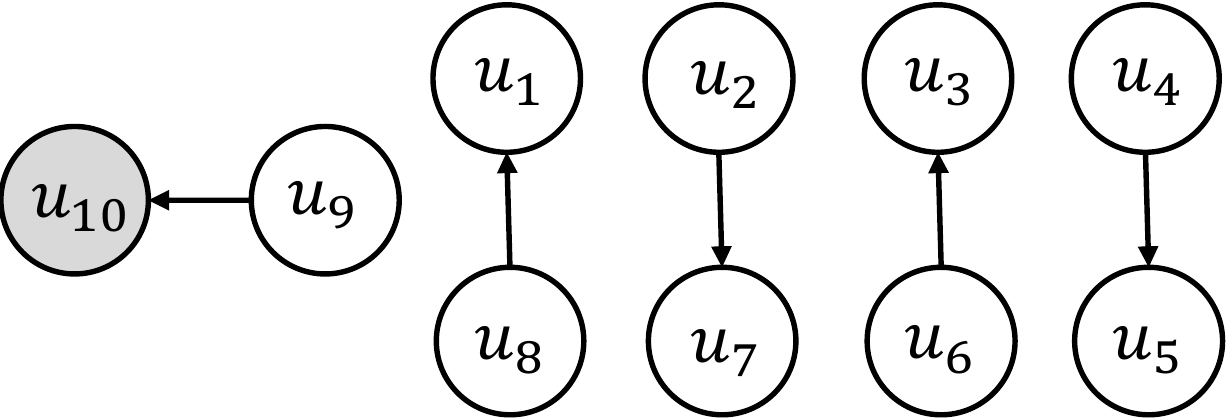}
    \caption{The super-game schedule in the first time slot for an instance with $m=10$}
    \label{figa1}
\end{figure}

From Figure~\ref{figa1}, we can see that the super-team $u_m=u_{10}$ is denoted as a dark node and all other super-teams $u_1, \dots, u_{m-1}$ are denoted as white nodes.
The white nodes form a cycle $(u_1,u_2,\dots,u_{m-1}, u_1)$.
In the second time slot, we keep the position of $u_m$ unchanged, change the positions of white super-teams in the cycle by moving one position in the clockwise direction, and also change the direction of each edge except for the most left edge incident on $u_m$. Please see Figure~\ref{figa2} for an illustration of the schedule in the second time slot.

In the second time slot, the super-game including $u_m$ is a left super-game and we put a letter `L' on the edge in Figure~\ref{figa2} to indicate this.
All other super-games are still normal super-games. In the second time slot, there are $\frac{m}{2}-1$ normal super-games and one left super-game.

 \begin{figure}[ht]
    \centering
    \includegraphics[scale=0.55]{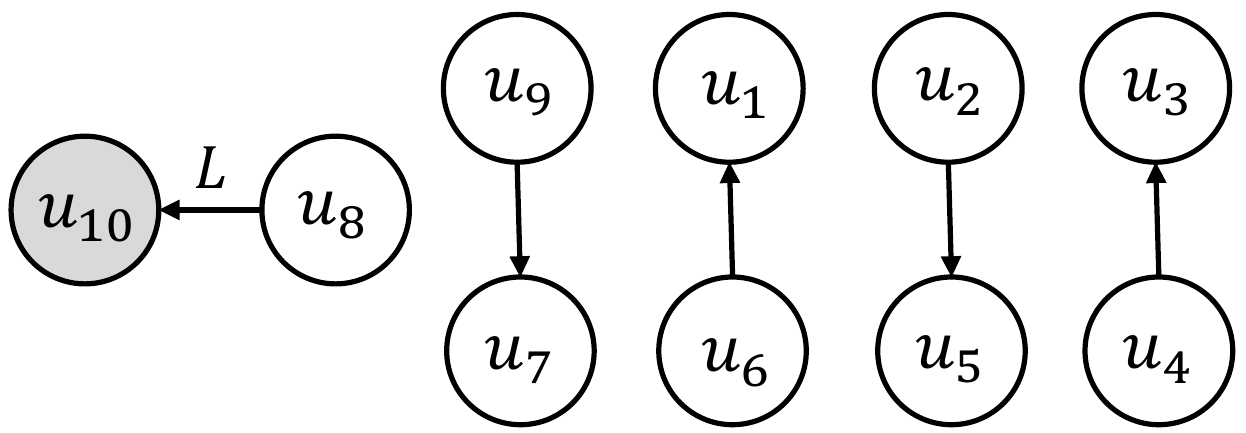}
    \caption{The super-game schedule in the second time slot for an instance with $m=10$}
    \label{figa2}
\end{figure}

In the third time slot,
we also change the positions of white super-teams in the cycle by moving one position in the clockwise direction while the direction of all edges is reversed. The position of the dark node $u_m$ will always keep the same.
In this time slot, there are still $\frac{m}{2}-1$ normal super-games and one left super-game that contains the super-team $u_m$. An illustration of the schedule in the third time slot is shown in Figure~\ref{figa3}.

\begin{figure}[ht]
    \centering
    \includegraphics[scale=0.55]{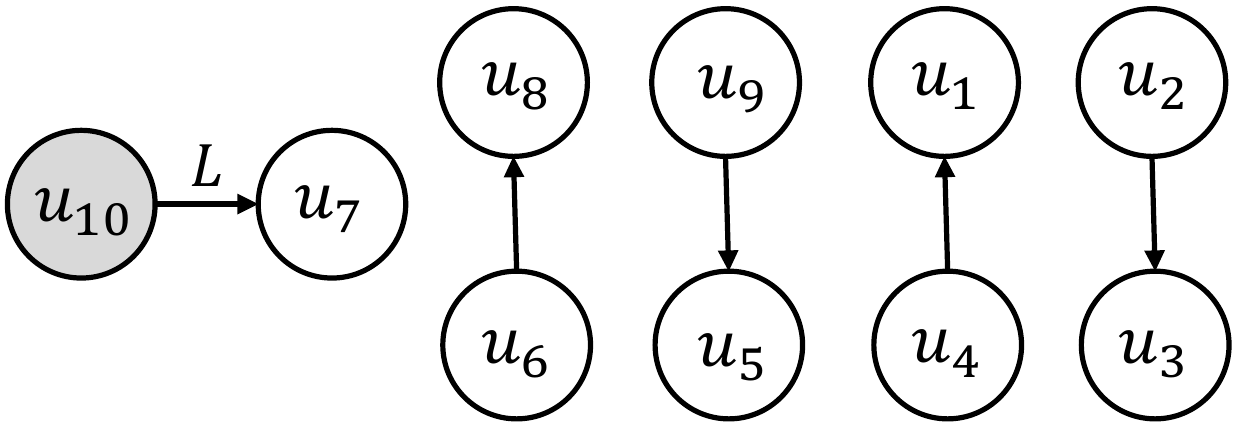}
    \caption{The super-game schedule in the third time slot for an instance with $m=10$}
    \label{figa3}
\end{figure}

The schedules for the other time slots are derived analogously.
However, the kinds of super-games in different time slots may be different.
For the first time slot, all the $\frac{m}{2}$ super-games in it are normal super-games.
For time slot $i$ $(2\leq i \leq m-3)$, the super-game involving super-team $u_m$ is a left super-game and all other super-games are normal super games.
For time slot $m-2$, all the $\frac{m}{2}$ super-games in it are penultimate super-games.
For time slot $m-1$, all the $\frac{m}{2}$ super-games in it are last super-games.  Figure~\ref{figa45} shows an illustration of the super-game schedule in the last two time slots, where we put a letter `P' (resp., `T') on the edge to indicate that the super-game is a penultimate (resp., last) super-game.

\begin{figure}[ht]
    \centering
    \includegraphics[scale=0.55]{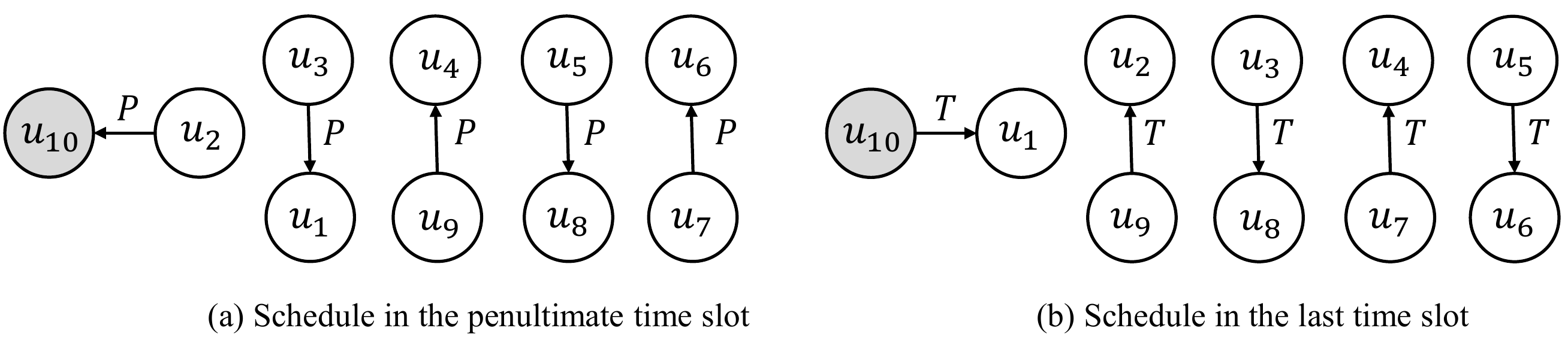}
    \caption{The super-game schedules in the last two time slots for an instance with $m=10$}
    \label{figa45}
\end{figure}

Next, we explain how to extend the super-games to normal games.
Recall that we have four kinds of super-games: normal, left, penultimate, and last.

\textbf{Case~1. Normal super-games}:
Each normal super-game will be extended to eight normal games on four consecutive days.
Assume that in a normal super-game, super-team $u_{i}$ plays against the super-team $u_{j}$ at the home venue of $u_j$ in time slot $q$ ($1\leq i,j\leq m$ and $1\leq q\leq m-3$). Recall that $u_{i}$ represents normal teams \{$t_{2i-1}, t_{2i}$\} and $u_{j}$ represents normal teams \{$t_{2j-1}, t_{2j}$\}. The super-game will be extended to eight normal games on four corresponding days from $4q-3$ to $4q$, as shown in Figure~\ref{figa6}. A directed edge from team $t_{i'}$ to team $t_{i''}$ means that $t_{i'}$ plays against $t_{i''}$ at the home venue of $t_{i''}$.
Note that if the super-game is at the home venue of $u_i$, i.e., there is directed edge from $u_j$ to $u_i$, then the direction of all edges in the figure will be reversed.

\begin{figure}[ht]
    \centering
    \includegraphics[scale=1]{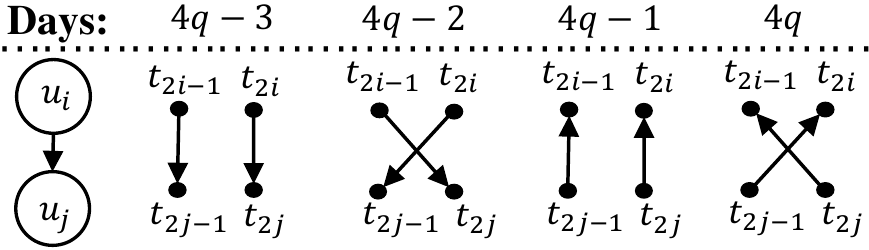}
    \caption{Extending normal super-games}
    \label{figa6}
\end{figure}

\textbf{Case~2. Left super-games}:
Assume that in a left super-game, super-team $u_{i}$ plays against super-team $u_{m}$ at the home venue of $u_m$ in (even) time slot $q$ ($3\leq i\leq m-2$ and $2\leq q\leq m-3$). Recall that $u_{m}$ represents normal teams \{$t_{2m-1}, t_{2m}$\} and $u_{i}$ represents normal teams \{$t_{2i-1}, t_{2i}$\}. The super-game will be extended to eight normal games on four corresponding days from $4q-3$ to $4q$, as shown in Figure~\ref{figa7}, for even time slot $q$. Note that the direction of edges in the figure will be reversed for odd time slot $q$.

\begin{figure}[ht]
    \centering
    \includegraphics[scale=1]{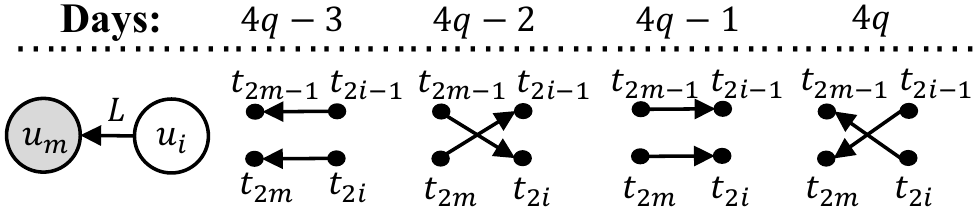}
    \caption{Extending left super-games}
    \label{figa7}
\end{figure}

\textbf{Case~3. Penultimate super-games}:
Assume that in a penultimate super-game, super-team $u_{i}$ plays against super-team $u_{j}$ at the home venue of $u_j$ in time slot $q=m-2$ ($1\leq i,j\leq m$). The super-game will be extended to eight normal games on four corresponding days from $4m-11$ to $4m-8$, as shown in Figure~\ref{figa8}.

\begin{figure}[ht]
    \centering
    \includegraphics[scale=1]{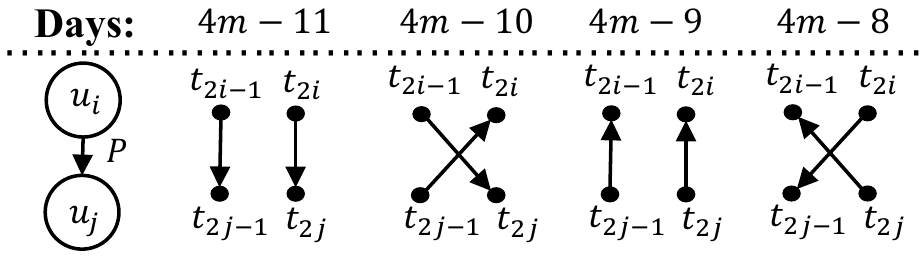}
    \caption{Extending penultimate super-games}
    \label{figa8}
\end{figure}

\textbf{Case~4. Last super-games}:
Assume that in a last super-game, super-team $u_{i}$ plays against super-team $u_{j}$ at the home venue of $u_j$ in time slot $q=m-1$ ($1\leq i,j\leq m$). The super-games will be extended to twelve normal games on six corresponding days from $4m-7$ to $4m-2$, as shown in Figure~\ref{figa9}.

\begin{figure}[ht]
    \centering
    \includegraphics[scale=1]{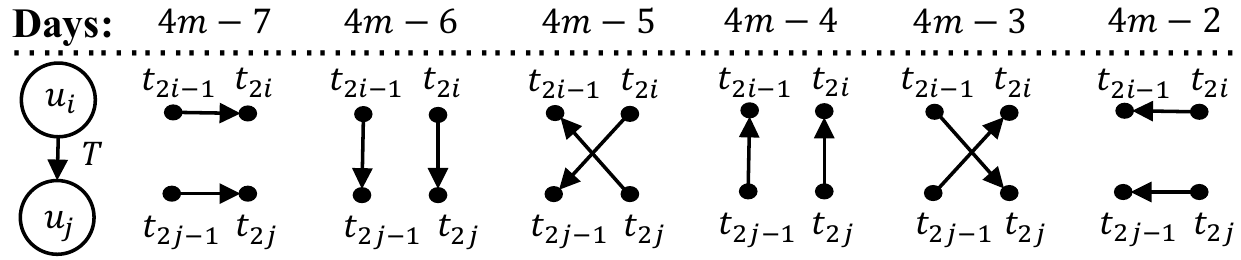}
    \caption{Extending last super-games}
    \label{figa9}
\end{figure}

We have described the main part of the scheduling algorithm. Before proving its feasibility, we first show an example of the schedule for $n=8$ teams constructed by our method. In Table~\ref{ansexample}, the $i$-th row indicates team $t_i$, the $j$-th column indicates the $j$-th day in the double round-robin, item $+t_{x}$ (resp., $-t_x$) on the $i$-th row and $j$-th column means that team $t_i$ plays against team $t_{x}$ in the $j$-th day at the home venue of team $t_{x}$ (resp., $t_i$).

\begin{table}[ht]
\centering
\begin{tabular}{c|cccccccccccccc}
   & 1 & 2 & 3 & 4 & 5 & 6 & 7 & 8 & 9 & 10 & 11 & 12 & 13 & 14\\
   \hline
  $t_{1}$ & $-t_{3}$ & $-t_{4}$ & $+t_{3}$ & $+t_{4}$ & $-t_{5}$ & $+t_{6}$ & $+t_{5}$ & $-t_{6}$ &
            $+t_{2}$ & $-t_{7}$ & $-t_{8}$ & $+t_{7}$ & $+t_{8}$ & $-t_{2}$ \\
  $t_{2}$ & $-t_{4}$ & $-t_{3}$ & $+t_{4}$ & $+t_{3}$ & $-t_{6}$ & $-t_{5}$ & $+t_{6}$ & $+t_{5}$ &
            $-t_{1}$ & $-t_{8}$ & $+t_{7}$ & $+t_{8}$ & $-t_{7}$ & $+t_{1}$ \\
  $t_{3}$ & $+t_{1}$ & $+t_{2}$ & $-t_{1}$ & $-t_{2}$ & $+t_{7}$ & $+t_{8}$ & $-t_{7}$ & $-t_{8}$ &
            $+t_{4}$ & $-t_{5}$ & $-t_{6}$ & $+t_{5}$ & $+t_{6}$ & $-t_{4}$ \\
  $t_{4}$ & $+t_{2}$ & $+t_{1}$ & $-t_{2}$ & $-t_{1}$ & $+t_{8}$ & $-t_{7}$ & $-t_{8}$ & $+t_{7}$ &
            $-t_{3}$ & $-t_{6}$ & $+t_{5}$ & $+t_{6}$ & $-t_{5}$ & $+t_{3}$ \\
  $t_{5}$ & $+t_{7}$ & $+t_{8}$ & $-t_{7}$ & $-t_{8}$ & $+t_{1}$ & $+t_{2}$ & $-t_{1}$ & $-t_{2}$ &
            $+t_{6}$ & $+t_{3}$ & $-t_{4}$ & $-t_{3}$ & $+t_{4}$ & $-t_{6}$ \\
  $t_{6}$ & $+t_{8}$ & $+t_{7}$ & $-t_{8}$ & $-t_{7}$ & $+t_{2}$ & $-t_{1}$ & $-t_{2}$ & $+t_{1}$ &
            $-t_{5}$ & $+t_{4}$ & $+t_{3}$ & $-t_{4}$ & $-t_{3}$ & $+t_{5}$ \\
  $t_{7}$ & $-t_{5}$ & $-t_{6}$ & $+t_{5}$ & $+t_{6}$ & $-t_{3}$ & $+t_{4}$ & $+t_{3}$ & $-t_{4}$ &
            $+t_{8}$ & $+t_{1}$ & $-t_{2}$ & $-t_{1}$ & $+t_{2}$ & $-t_{8}$ \\
  $t_{8}$ & $-t_{6}$ & $-t_{5}$ & $+t_{6}$ & $+t_{5}$ & $-t_{4}$ & $-t_{3}$ & $+t_{4}$ & $+t_{3}$ &
            $-t_{7}$ & $+t_{2}$ & $+t_{1}$ & $-t_{2}$ & $-t_{1}$ & $+t_{7}$ \\
\end{tabular}
\caption{The schedule for $n=8$ teams, where the horizontal ordinate represents the teams, the ordinate represents the days, and
`$+$' (resp., `$-$') means that the team on the corresponding horizontal ordinate plays at its opponent's home (resp., own home)}
\label{ansexample}
\end{table}

From Table~\ref{ansexample}, we can roughly check the feasibility of this instance: on each line there are at most two consecutive `$+$'/`$-$', and each team plays the required games.
Next, we formally prove the correctness of our algorithm.

\begin{theorem}\label{feas1}
For TTP-$2$ with $n$ teams such that $n\equiv 0 \pmod 4$ and $n\geq 8$, the above construction can generate a feasible schedule.
\end{theorem}
\begin{proof}
By the definition of feasible schedules, we need to prove the five properties: fixed-game-value, fixed-game-time, direct-traveling, no-repeat, and bounded-by-$k$.

The first two properties -- fixed-game-value and fixed-game-time are easy to see.
Each super-game in the first $m-2$ time slots will be extended to eight normal games on four days and each team participates in four games on four days. Each super-game in the last time slot will be extended to twelve normal games on six days and each team participates in six games on six days. So each team plays $2(n-1)$ games on $2(n-1)$ different days. Since there is a super-game between each pair super-teams, it is also easy to see that each team pair plays exactly two games, one at the home venue of each team. For the third property, we assume that the itinerary obeys the direct-traveling property and it does not need to be proved.

It is also easy to see that each team will not violate the no-repeat property.
In any time slot, no two normal games between the same pair of normal teams are arranged on two consecutive days according to the ways to extend super-games to normal games. For two days of two different time slots, each super-team will play against a different super-team and then a normal team will also play against a different normal team.

Last, we prove that each team does not violate the bounded-by-$k$ property. We will use `$H$' and `$A$' to denote a home game and an away game, respectively. We will also let $\overline{H}=A$ and $\overline{A}=H$.

First, we look at the games in the first $m-3$ time slots, i.e., the first $2n-12$ days. For the two teams in $u_m$, the four games in the first time slot will be $HHAA$, in an even time slot will be $HAAH$ (see Figure~\ref{figa7}), and in an odd time slot (not including the first time slot) will be $AHHA$. So two consecutive time slots can combine well without creating three consecutive home/away games.

Next, we consider a team $t_i$ in $u_j$ $(j\in \{1,2,\dots,m-1\})$.
For a normal super-game involving $u_j$, if the direction of the edge (the normal super-game) is from $u_j$ to another super-team, then the corresponding four games including $t_i$ will be $AAHH$, and if the direction of the edge is reversed, then the corresponding four games including $t_i$ will be $\overline{AAHH}=HHAA$.
For a left super-game involving $u_j$, if the direction of the edge (the normal super-game) is from $u_j$ to another super-team, then the corresponding four games including $t_i$ will be $AHHA$, and if the direction of the edge is reversed, then the corresponding four games including $t_i$ will be $\overline{AHHA}=HAAH$.
Note that in the first $m-3$ time slots the direction of the edge incident on super-team $u_j$ will only change after the left super-game. So two consecutive time slots can combine well without creating three consecutive home/away games.

Finally, we consider the last ten days in the last two time slots (time slots $m-2$ and $m-1$).
For the sake of presentation, we let $t_{i_1}=t_{2i-1}$ and $t_{i_2}=t_{2i}$.
We just list out the last ten games in the last two time slots for each team, which are shown in Figure~\ref{figa10}.
There are four different cases for the last ten games: $u_o$, $u_e$, $u_2$, and $u_m$, where $o\in \{3,\dots, m-1\}$ and $e\in \{1\}\cup \{4,\dots, m-2\}$.

\begin{figure}[ht]
    \centering
    \includegraphics[scale=0.7]{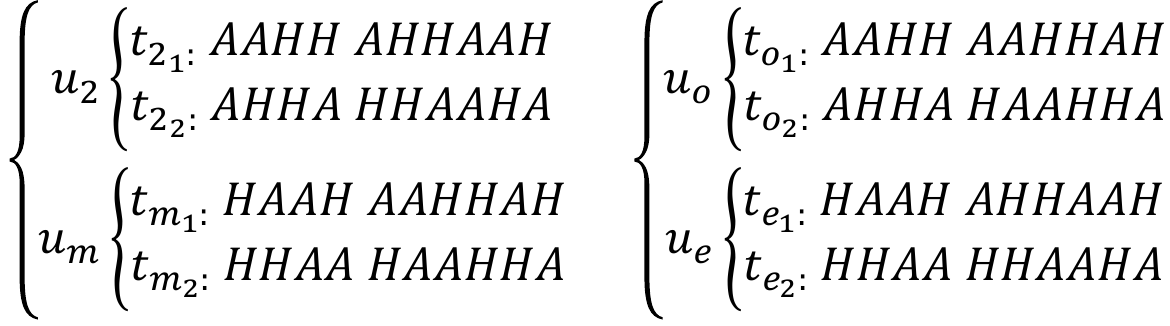}
    \caption{The last ten games for the case of $n\equiv 0 \pmod 4$, where $o\in \{3,\dots, m-1\}$ and $e\in \{1\}\cup \{4,\dots, m-2\}$}
    \label{figa10}
\end{figure}

From Figure~\ref{figa10}, we can see that there are no three consecutive home/away games in the last ten days. It is also easy to see that on day $2n-12$ (the last day in time slot $m-3$), the games for $t_{2_1}$ and $t_{2_2}$ (in $u_2$) are $H$, the games for $t_{m_1}$ and $t_{m_2}$ (in $u_m$) are $A$, the games for $t_{o_1}$ and $t_{o_2}$ (in $u_o$) are $H$, and the games for $t_{e_1}$ and $t_{e_2}$ (in $u_e$) are $A$. So time slots $m-3$ and $m-2$ can also combine
well without creating three consecutive home/away games.

Thus, the bounded-by-$k$ property also holds.
Since our schedule satisfies all the five properties of feasible schedules, we know our schedule is feasible for TTP-2.
\end{proof}

\subsection{Analyzing the Approximation Quality}
To show the quality of our schedule, we compare it with the independent lower bound. We will check the difference between our itinerary of each team $t_i$ and the optimal itinerary of $t_i$ and compute the expected extra cost.
As mentioned in the last paragraph of Section~\ref{sec_pre}, we will compare some sub itineraries of a team.
According to the construction,  we can see that all teams stay at home before the first game in a super-game and return home after the last game in the super-game.
Hence, we will look at the sub itinerary of a team on the four or six days in a super-game, which is coincident with a sub itinerary of the optimal itinerary.
In our algorithm, there are four kinds of super-games: normal super-games, left super-games, penultimate super-games and last super-games. We analyze the total expected extra cost of all normal teams caused by each kind of super-game.

\begin{lemma}\label{extra}
Assume there is a super-game between super-teams $u_i$ and $u_j$ at the home of $u_j$ in our schedule.
\begin{enumerate}
\item [(a)] If the super-game is a normal super-game, then the expected extra cost of all normal teams in $u_i$ and $u_j$ is 0;
\item [(b)] If the super-game is a left super-game, then the expected extra cost of all normal teams in $u_i$ and $u_j$ is at most $\frac{4}{n(n-2)}\LB$;
\item [(c)] If the super-game is a penultimate/last super-game, then the expected extra cost of all normal teams in $u_i$ and $u_j$ is at most $\frac{2}{n(n-2)}\LB$.
\end{enumerate}
\end{lemma}
\begin{proof}
From Figure~\ref{figa6} we can see that in a normal super-game any of the four normal teams will visit the home venue of the two normal teams in the opposite super-team in one road trip. So they have the same road trips as that in their optimal itineraries.
The extra cost is 0. So (a) holds.

Next, we assume that the super-game is a left super-game. We mark here that one super-team is $u_m$.
From Figure~\ref{figa7} we can see that the two teams in the super-team $u_i$ play $AHHA$ on the four days (Recall that $A$ means an away game and $H$ means a home game), and the two teams in the super-team $u_j$ play $HAAH$.
The two teams in $u_j$ will have the same road trips as that in the optimal itinerary and then the extra cost is 0.
We compare the road trips of the two teams $t_{2i-1}$ and $t_{2i}$ in $u_i$ with their coincident sub itineraries of their optimal itineraries. In our schedule, team $t_{2i-1}$ contains two road trips $(t_{2i-1},t_{2j-1},t_{2i-1})$ and $(t_{2i-1},t_{2j},t_{2i-1})$ while the coincident sub itinerary of its optimal itinerary contains one road trip $(t_{2i-1},t_{2j-1},t_{2j},t_{2i-1})$, and team $t_{2i}$ contains two road trips $(t_{2i},t_{2j},t_{2i})$ and $(t_{2i},t_{2j-1},t_{2i})$ while the coincident sub itinerary of its optimal itinerary contains one road trip $(t_{2i},t_{2j-1},t_{2j},t_{2i})$. The expected difference is
\[
\EE{D_{2i-1,2j-1}+D_{2i-1,2j}+D_{2i,2j-1}+D_{2i,2j}-2D_{2j-1,2j}}\leq\frac{4}{n(n-2)}\LB,
\]
since $\EE{D_{2i-1,2j-1}}=\EE{D_{2i-1,2j}}=\EE{D_{2i,2j-1}}=\EE{D_{2i,2j}}\leq\frac{1}{n(n-2)}\LB$ by Lemma~\ref{core}. So (b) holds.

Then, we suppose that the super-game is a penultimate super-game. From Figure~\ref{figa8} we can see that the two teams $t_{2i-1}$ and $t_{2i}$ in the super-team $u_i$ play $AAHH$ and $AHHA$, respectively, and the two teams $t_{2j-1}$ and $t_{2j}$ in the super-team $u_j$ play $HAAH$ and $HHAA$, respectively.
Teams $t_{2i-1}$, $t_{2j-1}$, and $t_{2j}$ will have the same road trip as that in their optimal itineraries and then the extra cost is 0.
We compare the road trips of team $t_{2i}$ with its optimal road trip. In our schedule, team $t_{2i}$ contains two road trips $(t_{2i},t_{2j},t_{2i})$ and $(t_{2i},t_{2j-1},t_{2i})$ while the coincident sub itinerary of its optimal itinerary contains one road trip $(t_{2i},t_{2j-1},t_{2j},t_{2i})$.
By Lemma~\ref{core}, we can get the expected difference is
\[
\EE{D_{2i,2j-1}+D_{2i,2j}-D_{2j-1,2j}}\leq \frac{2}{n(n-2)}\LB.
\]
Last, we consider it as a last super-game.
From Figure~\ref{figa9} we can see that the two teams $t_{2i-1}$ and $t_{2i}$ in the super-team $u_i$ play $AAHHAH$ and $HAAHHA$, respectively, and the two teams $t_{2j-1}$ and $t_{2j}$ in the super-team $u_j$ play $AHHAAH$ and $HHAAHA$, respectively.
Teams $t_{2i}$, $t_{2j-1}$, and $t_{2j}$ will have the same road trips as that in their optimal itineraries and then the extra cost is 0.
We compare the road trips of team $t_{2i-1}$ with its optimal road trips. In our schedule, team $t_{2i-1}$ contains two road trips $(t_{2i-1},t_{2i},t_{2j-1},t_{2i-1})$ and $(t_{2i-1},t_{2j},t_{2i-1})$ while the coincident sub itinerary of its optimal itinerary contains two road trip $(t_{2i-1},t_{2j-1},t_{2j},t_{2i-1})$ and $(t_{2i-1},t_{2i},t_{2i-1})$.
By Lemma~\ref{core}, the team $t_{2i-1}$ in $u_i$ has an expected extra cost of
\[
\EE{D_{2i-1,2j}+D_{2i,2j-1}-D_{2i-1,2i}-D_{2j-1,2j}}\leq \frac{2}{n(n-2)}\LB.
\]
So (c) holds.
\end{proof}

In our schedule, there are $\frac{m}{2}+(m-4)(\frac{m}{2}-1)$ normal super-games, which contribute 0 to the expected extra cost.
There are $m-4$ left super-games in the $m-4$ middle time slots, $\frac{m}{2}$ penultimate super-games in the penultimate time slot, and $\frac{m}{2}$ last super-games in the last time slot.
By Lemma~\ref{extra}, we know that the total expected extra cost is
\[
(m-4)\times\frac{4}{n(n-2)}\LB+\lrA{\frac{m}{2}+\frac{m}{2}}\times\frac{2}{n(n-2)}\LB=\lrA{\frac{3}{n}-\frac{10}{n(n-2)}}\LB.
\]

Dominated by the time of computing a
minimum weight perfect matching~\cite{gabow1974implementation,lawler1976combinatorial}, the running time of the algorithm is $O(n^3)$. We can get that

\begin{theorem}\label{res_1}
For TTP-$2$ with $n$ teams, when $n\geq 8$ and $n\equiv 0 \pmod 4$, there is a randomized $O(n^3)$-time algorithm with an expected approximation ratio of $1+{\frac{3}{n}}-{\frac{10}{n(n-2)}}$.
\end{theorem}

Next, we show that the analysis of our algorithm is tight, i.e., the approximation ratio in Theorem~\ref{res_1} is the best for our algorithm. We show an example that the ratio can be reached.

In the example, the distance of each edge in the minimum perfect matching $M$ is 0 and the distance of each edge in $G-M$ is 1.
We can see that the triangle inequality property still holds.
By (\ref{eqn_lowerbound}), we know that the independent lower bound of this instance is
\[
\LB=2D_G+nD_M=2\size{G-M}=n(n-2).
\]

In this instance, the extra costs of a normal super-game, left super-game, penultimate super-game and last super-game are 0, 4, 2 and 2, respectively.
In our schedule, there are $m-4$ left super-games, $\frac{m}{2}$ penultimate super-games and $\frac{m}{2}$ last super-games in total.
Thus, the total extra cost of our schedule is  $(m-4)\times4+(\frac{m}{2}+\frac{m}{2})\times2=3n-16$. Thus, the ratio is
\[
1+\frac{3n-16}{n(n-2)}=1+\frac{3}{n}-\frac{10}{n(n-2)}.
\]

This example only shows the ratio is tight for this algorithm. However, it is still possible that some other algorithms can achieve a better ratio.

\subsection{Techniques for Further Improvement}
From Lemma~\ref{extra}, we know that there are three kinds of super-games: left super-games, penultimate super-games and last super-games that can cause extra cost. If we can reduce the total number of these super-games, then we can improve our schedule.
Based on this idea, we will apply divide-and-conquer to our packing-and-combining method which can reduce the total number of left super-games for some cases.

We first extend the packing-and-combining method. Suppose $n\geq 8p$ and $n\equiv 0 \pmod {4p}$, then we pack $p$ super-teams as a \emph{group-team}. There are $g=\frac{n}{2p}$ group-teams. Since the previous construction corresponds to the case $p=1$, we assume $p\geq 2$ here.

Similar to the previous construction, there are $g-1$ group-slots and each group-slot contains $\frac{g}{2}$ \emph{group-games}. But, there are only three kinds of group-games: normal, left, and last.
In the first group-slot, all group-games are normal group-games. In the middle $g-3$ group-slots, there is always one left group-game and $\frac{g}{2}-1$ normal group-games. In the last group-game, there are $\frac{g}{2}$ last group-games.

Next, we show how to extend these three kinds of group-games.
We assume that there is a group-game between two group-teams $U_i$ and $U_j$ at the home of $U_j$, where $U_i=\{u_{i_1},\dots,u_{i_p}\}$ and $U_j=\{u_{j_1},\dots,u_{j_p}\}$.

\textbf{Normal group-games:}
For a normal group-game, we extend it into $p$ time slots where in the $l$-th time slot, there are $p$ normal super-games: $\{u_{i_{i'}}\rightarrow u_{j_{(i'+l-2 \bmod p)+1}}\}_{i'=1}^{p}$.
Hence, there are $p^2$ normal super-games in a normal group-game. According the design of normal super-games, all games between group-teams $U_i$ and $U_j$ (i.e., between one team in a super-team in $U_i$ and one team in a super-team in $U_j$) are arranged.
Note that all super-teams in $U_i$ (resp., $U_j$) play away (resp., home) super-games.

\textbf{Left group-games:}
For a left group-game, we also extend it into $p$ time slots and the only difference with the normal group-game is that the $p$ super-games of the left group-game in the $p$-th time slot are $p$ left super-games. Hence, there are $p$ left super-games and $p(p-1)$ normal super-games in a left group-game.
According the design of normal/left super-games, all games between group-teams $U_i$ and $U_j$ are arranged.
Note that all super-teams in $U_i$ (resp., $U_j$) play away (resp., home) super-games.

\textbf{Last group-games:}
For a last group-game, we need to arrange all games inside each group-team as well as between these two group-teams, which form a double round-robin for teams in these two group-teams. Similar with the normal and left group-games, the super-teams in $U_i$ are ready to play away normal super-games while the super-teams in $U_j$ are ready to play home normal super-games.
Both of the general construction ($p\geq2$) and the previous construction ($p=1$) start with normal super-games.
Thus, the last group-game between $U_i$ and $U_j$ can be seen as a sub-problem of our construction with $4p$ teams, which can be seen as a divided-and-conquer method. Note that in this sub-problem we need to make sure that the super-teams in $U_i$ start with away super-games and the super-teams in $U_j$ start with home super-games.

Before the last group-slot, there are $(g-3)\times p=\frac{n}{2}-3p$ left super-games in total.
In the last group-slot, there are $\frac{g}{2}=\frac{n}{4p}$ last group-games. Therefore, if $n\geq 8p$ and $n\equiv 0 \pmod{4p}$, then the minimum total number of left super-games in our construction of packing $p$ super-teams, denoted by $L_p(n)$, satisfies that

\begin{eqnarray}\label{formula}
L_p(n)=\frac{n}{2}-3p+\frac{n}{4p}\min_{1\leq i< p} L_{i}(4p).
\end{eqnarray}

Since the framework of our general construction is the same as our initial construction, the correctness is easy to observe. Here, we show an example of our general construction on an instance of $n=16$ and $p=2$ in Figure~\ref{figa11}.

\begin{figure}[ht]
    \centering
    \includegraphics[scale=0.55]{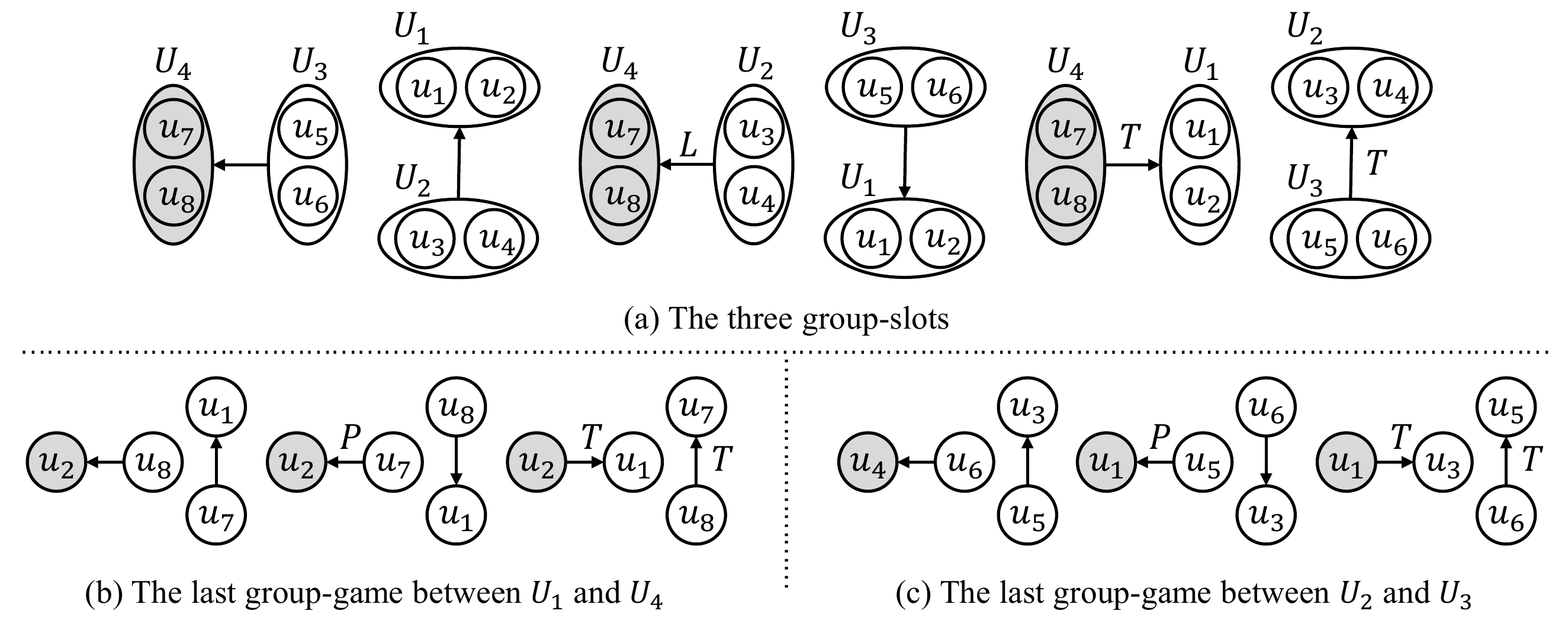}
    \caption{
    An illustration of our general construction on an instance of $n=16$ and $p=2$: the general construction contains $g-1=\frac{n}{2p}-1=3$ group-slots as shown in (a); the instance in the last time slot contains $\frac{g}{2}=\frac{n}{4p}=2$ last group-games, which can be seen two sub-problems of TTP-2 with $4p=8$ teams; each sub-problem uses the initial construction as shown in (b) and (c)}
    \label{figa11}
\end{figure}

Since the sub-problem will always reduce to the case $p=1$, the total number of penultimate/last super-games keeps unchanged. Hence, by the previous analysis, the total expected extra cost of $m/2$ penultimate super-games and $m/2$ last super-games is still $(\frac{m}{2}+\frac{m}{2})\times\frac{2}{n(n-2)}\LB=\frac{n}{n(n-2)}\LB$.
To analyze the approximation quality of this general construction, we only need to compute the number of left super-games used, which depends on the value $p$. The minimum number of left super-games, denoted by $L(n)$, satisfies that
\[
L(n)=\min_{\substack{p: n\geq 8p,\\n\bmod {4p}=0 }}L_p(n).
\]
By Lemma~\ref{core}, the expected extra cost of one left super-game is $\frac{4}{n(n-2)}\LB$.
Therefore, the expected extra cost of $L(n)$ left super-games is $\frac{4L(n)}{n(n-2)}\LB$. The expected approximation ratio of our algorithm is
\[
1+\frac{4L(n)+n}{n(n-2)}.
\]
Note that the value $L(n)$ can be computed in $O(n^2)$ by using a dynamic programming method.

\begin{theorem}
For TTP-$2$ with $n$ teams, there is a randomized $O(n^3)$-time algorithm with an expected approximation ratio of $1+\frac{4L(n)+n}{n(n-2)}$, where $L(n)\leq \frac{n}{2}-4$ for $n\geq 8$ and $n\equiv 0 \pmod 4$, and $L(n)\leq \frac{n}{2}-6$ for $n\geq 16$ and $n\equiv 0 \pmod 8$, i.e., the expected approximation ratio is at most $1+\frac{3}{n}-\frac{10}{n(n-2)}$ for the former case and $1+\frac{3}{n}-\frac{18}{n(n-2)}$ for the latter case.
\end{theorem}
\begin{proof}
According to our initial construction, we have $L(n)\leq L_1(n)=\frac{n}{2}-4$ for $n\geq 8$ and $n\equiv 0\pmod 4$. Thus, we have $L_1(8)=0$, and we can get $L(n)\leq L_2(n)=\frac{n}{2}-6$ for $n\geq 16$ and $n\equiv 0\pmod 8$.
For $n\geq 16$ and $n\equiv 0\pmod 8$, the general construction can reduce at least two left super-games. Hence, for this case, the expected approximation ratio of our algorithm is at most $1+\frac{3}{n}-\frac{18}{n(n-2)}$.
\end{proof}

We can reduce more than two left super-games for some other cases. For example, we can get $L(32)=L_4(32)=8$ and $L(64)=L_8(64)=24$, which has four less left super-games compared with $L_1(32)=12$ and $L_1(64)=28$.
Since the biggest instance on the benchmark is 40, we show an example on the number of left super-games before and after using divide-and-conquer (D\&C) for $n$ varying from $8$ to $40$ in Table~\ref{number}. It is worth noting that when $n>64$ our programming shows that the reduced number of left super-games is at most 2. Hence, we conjecture that $L(n)=L_2(n)=\frac{n}{2}-6$ for $n>64$ and $n\equiv 0\pmod 8$.
As $n$ goes larger, we can not reduce more left super-games.

\begin{table}[ht]
\centering
\begin{tabular}{c|cc}
\hline
  Data & Before & After\\
   size & D\&C & D\&C\\
\hline
  $n=40$ & 16 & $\mathbf{14}$\\
  $n=36$ & 14 & 14\\
  $n=32$ & 12 & $\mathbf{8}$ \\
  $n=28$ & 10 & 10\\
  $n=24$ & 8  & $\mathbf{6}$ \\
  $n=20$ & 6  & 6 \\
  $n=16$ & 4  & $\mathbf{2}$ \\
  $n=12$ & 2  & 2 \\
  $n=8$  & 0  & 0
\end{tabular}
\caption{Results on the number of left super-games}
\label{number}
\end{table}

The generalized algorithm does not work for $n\equiv 4\pmod 8$.
In Table~\ref{number}, we can see that the improved number of left super-games for $n=32$ (bigger case) is even less than the number for $n=28$ (smaller case).
We conjecture that there also exist a method to reduce the number of left super-games for $n\equiv 4\pmod 8$.

\section{The Construction for Odd $n/2$}
\subsection{Construction of the Schedule}
When $n/2$ being odd, the construction will be slightly different for $n\equiv 2 \pmod 8$ and $n\equiv 6 \pmod 8$. We will describe the algorithm for the case of $n\equiv 2 \pmod 8$. For the case of $n\equiv 6 \pmod 8$, only some edges will have different directions (games taking in the opposite place). These edges will be denoted by dash lines and we will explain them later.

In the schedule, there will be two special super-teams $u_{m-1}$ and $u_m$. For the sake of presentation, we will denote $u_l = u_{m-1}$ (resp., $u_r=u_{m}$), and denote the two teams in $u_l$ as $\{t_{l1}, t_{l2}\}$ (resp., the two teams in $u_r$ as $\{t_{r1}, t_{r2}\}$).

We first design the super-games between super-teams in the first $m-2$ time slots, and then consider super-games in the last time slot.
In each of the first $m-2$ time slots, we have $\frac{m-1}{2}$ super-games (note that $m$ is odd), where one super-game involves three super-teams and all other super-games involve two super-teams.

In the first time slot, the $\frac{m-1}{2}$ super-games are arranged as shown in Figure~\ref{figb1}.
The most right super-game involving $u_r$ is called the \emph{right super-game}. The right super-game is the only super-game involving three super-teams.
The other $\frac{m-1}{2}-1$ super-games are called \emph{normal super-games}. There are also directed edges in the super-games, which will be used to extend super-games to normal games.

\begin{figure}[ht]
    \centering
    \includegraphics[scale=0.55]{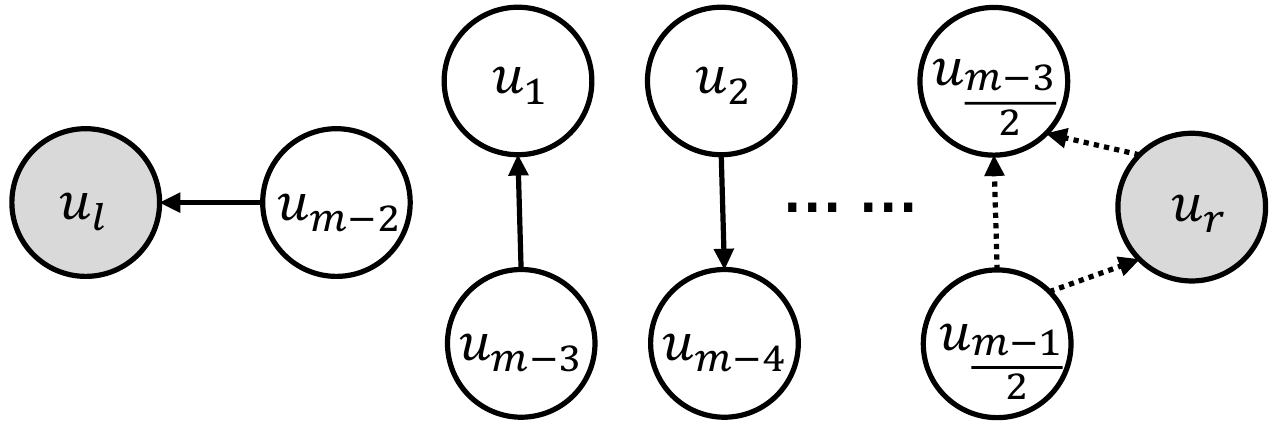}
    \caption{The super-game schedule in the first time slot}
    \label{figb1}
\end{figure}

Note that the dash edges in Figure~\ref{figb1} will be reversed for the case of $n\equiv 6\pmod 8$. We can also observe that the white nodes (super-teams $u_1, \dots, u_{m-2}$) in Figure~\ref{figb1} form a cycle $(u_1, u_2, \dots ,u_{m-2},u_1)$.
In the second time slot, super-games are scheduled as shown in Figure~\ref{figb2}. We change the positions of white super-teams in the cycle by moving one position in the clockwise direction, and also change the direction of each edge except for the most left edge incident on $u_l$. The super-game including $u_l$ is called a \emph{left super-game}. So in the second time slot, there is one left super-game, $\frac{m-1}{2}-2$ normal super-games and one right super-game.

\begin{figure}[ht]
    \centering
    \includegraphics[scale=0.55]{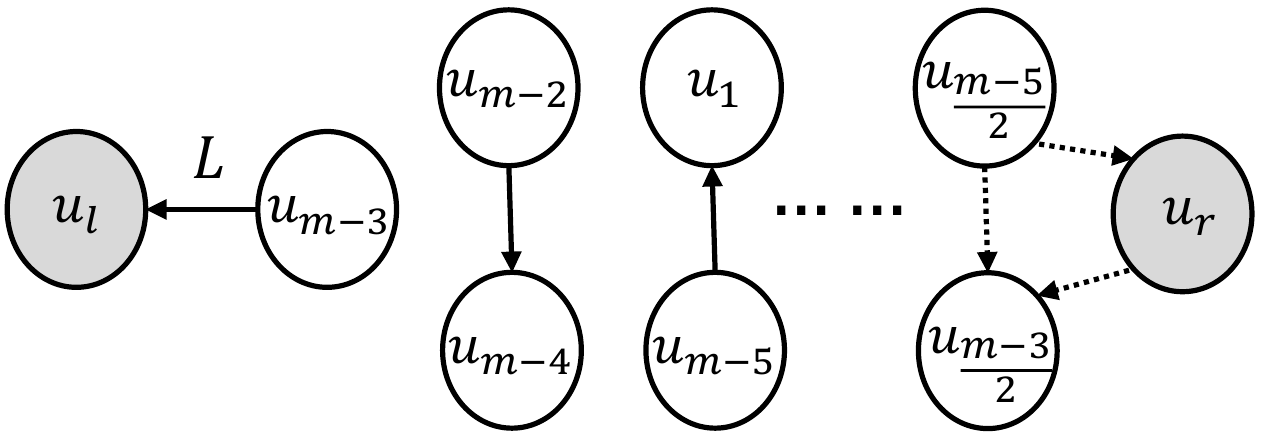}
    \caption{The super-game schedule in the second time slot}
    \label{figb2}
\end{figure}

In the third time slot, there is also one left super-game, $\frac{m-1}{2}-2$ normal super-games and one right super-game.
We also change the positions of white super-teams in the cycle by moving one position in the clockwise direction while the direction of each edge is reversed. The positions of the dark nodes will always keep the same.
An illustration of the schedule in the third time slot is shown in Figure~\ref{figb2+}.

\begin{figure}[ht]
    \centering
    \includegraphics[scale=0.55]{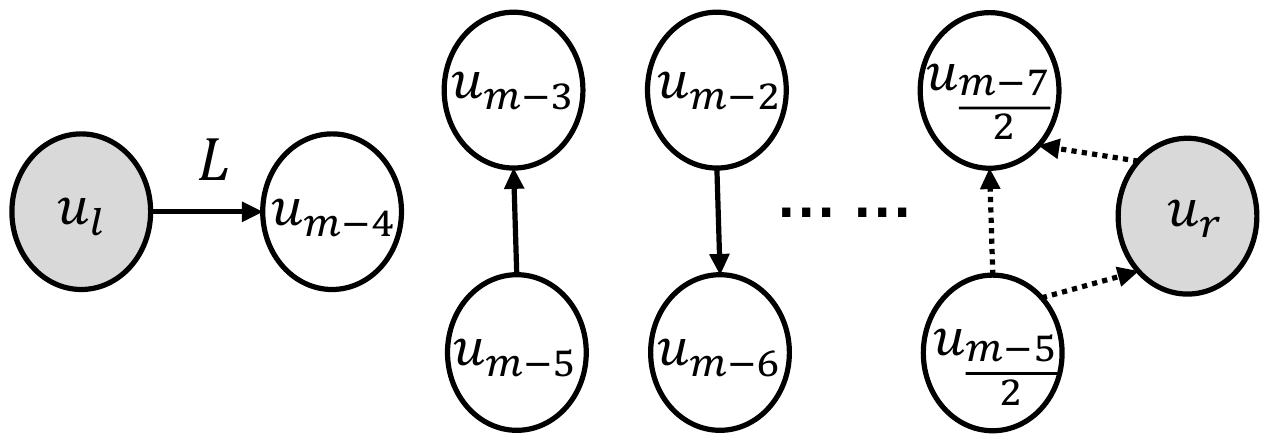}
    \caption{The super-game schedule in the third time slot}
    \label{figb2+}
\end{figure}

The schedules for the other middle time slots are derived analogously, however, in the time slot $m-2$, the left super-game will be special and we will explain it later. Next, we show how to extend the super-games in these time slots to normal games.

\textbf{Case~1. Normal super-games}:
We first consider normal super-games, each of which will be extended to four normal games on four days.
Assume that in a normal super-game, super-team $u_{i}$ plays against the super-team $u_{j}$ at the home venue of $u_j$ in time slot $q$ ($1\leq i,j \leq m-1$ and $1\leq q\leq m-2$). Recall that  $u_{i}$ represents normal teams $\{t_{2i-1}, t_{2i}\}$ and $u_{j}$ represents normal teams \{$t_{2j-1}, t_{2j}$\}. The super-game will be extended to eight normal games on four corresponding days from $4q-3$ to $4q$, as shown in Figure~\ref{figb3}. There is no difference compared with the normal super-games in Figure~\ref{figa6} for the case of even $n/2$.

\begin{figure}[ht]
    \centering
    \includegraphics[scale=1]{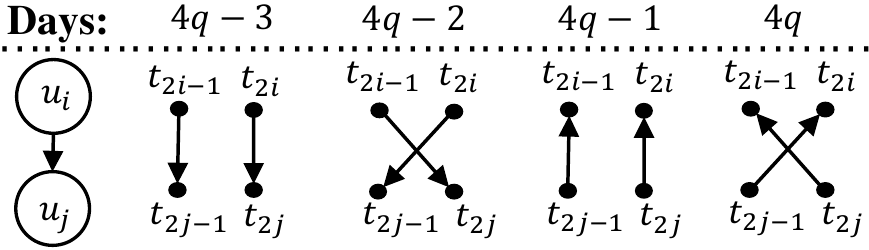}
    \caption{Extending normal super-games}
    \label{figb3}
\end{figure}

\textbf{Case~2. Left super-games}:
Assume that in a left super-game, super-team $u_{l}$ plays against the super-team $u_{i}$ at the home venue of $u_l$ in (even) time slot $q$ ($1\leq i\leq m-3$ and $2\leq q\leq m-2$). There are $m-3$ left super-games: the first $m-4$ left super-games in time slot $q$ with $2\leq q<m-2$ are (normal) left super-games; the left super-game in time slot $q$ with $q=m-2$ is a (special) left super-game.

We first consider normal left super-games.
Recall that $u_{l}$ represents normal teams \{$t_{l1}, t_{l2}$\} and $u_{i}$ represents normal teams \{$t_{2i-1}, t_{2i}$\}.
The super-game will be extended to eight normal games on four corresponding days from $4q-3$ to $4q$, as shown in Figure~\ref{figb4}, for even time slot $q$. Note that the direction of edges in the figure will be reversed for odd time slot $q$. There is no difference compared with the left super-games in Figure~\ref{figa7} for the case of even $n/2$.

\begin{figure}[ht]
    \centering
    \includegraphics[scale=1]{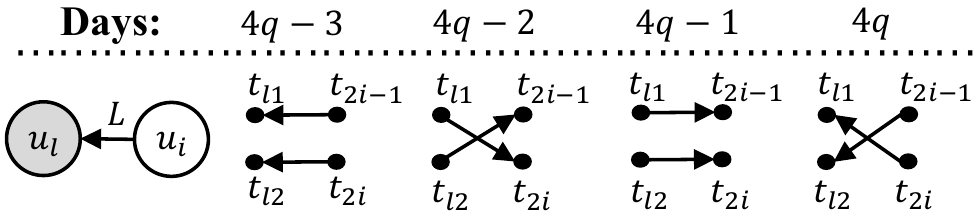}
    \caption{Extending left super-games}
    \label{figb4}
\end{figure}

In time slot $q=m-2$, there is a directed arc from super-team $u_l$ to super-team $u_1$. The super-game will be extended to eight normal games on four corresponding days from $4q-3$ to $4q$, as shown in Figure~\ref{figb4+}, where we put a letter `S' on the edge to indicate that the super-game is a (special) left super-game. There is no difference compared with the penultimate super-games in Figure~\ref{figa8} for the case of even $n/2$.

\begin{figure}[ht]
    \centering
    \includegraphics[scale=1]{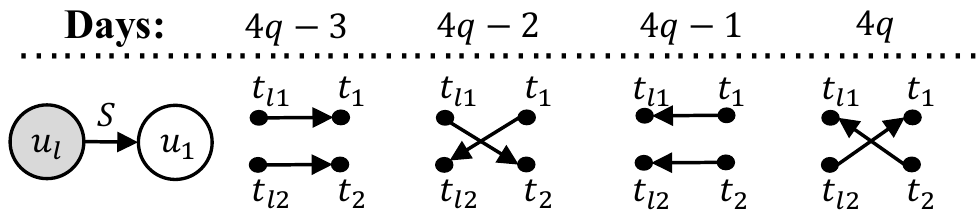}
    \caption{Extending the left super-game in the time slot $q=m-2$}
    \label{figb4+}
\end{figure}

\textbf{Case~3. Right super-games}:
Assume that in a right super-game, there are three super-teams $u_{i-1}$, $u_i$, and $u_{r}$ in time slot $q$ ($1\leq i\leq m-2$ and $1\leq q\leq m-2$) (we let $u_{0}=u_{m-2}$). Recall that $u_{r}$ represents normal teams \{$t_{r_1}, t_{r_2}$\}, $u_{i-1}$ represents normal teams \{$t_{2i-3}, t_{i-2}$\}, and $u_{i}$ represents normal teams \{$t_{2i-1}, t_{2i}$\}. The super-game will be extended to twelve normal games on four corresponding days from $4q-3$ to $4q$. Before extending the right super-games, we first introduce four types of right super-games $\{R_1,R_2,R_3,R_4\}$, as shown in Figure~\ref{figb5}. If the directions of the edges in $R_j$ are reversed, then we denote this type as $\overline{R_j}$.

\begin{figure}[ht]
    \centering
    \includegraphics[scale=1]{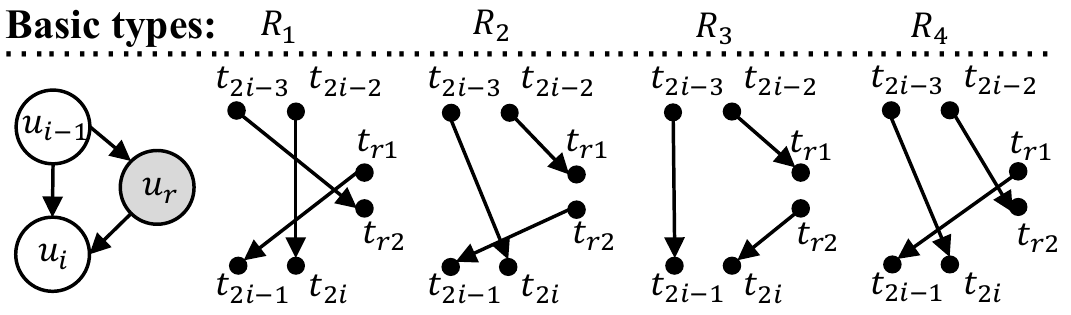}
    \caption{Four basic types of right super-games}
    \label{figb5}
\end{figure}

To extend right super-games, we consider three cases: $1\leq q\leq \frac{m-3}{2}$, $q=\frac{m-1}{2}$ and $\frac{m+1}{2}\leq q\leq m-2$.
For the case of $1\leq q\leq \frac{m-3}{2}$, if there is a directed edge from $u_{i-1}$ to $u_i$ (we know that $i-1$ is even and $i$ is odd), the four days of extended normal games are arranged as $R_4\cdot R_3\cdot \overline{R_4}\cdot \overline{R_3}$, otherwise, $\overline{R_2}\cdot \overline{R_1}\cdot R_2\cdot R_1$.
For the case of $q=\frac{m-1}{2}$, we know that there is a directed edge from $u_{m-2}$ to $u_1$ (both $m-1$ and $1$ are odd), the four days of extended normal games are arranged as $R_1\cdot R_3\cdot \overline{R_1}\cdot \overline{R_3}$.
For the case of $\frac{m+1}{2}\leq q\leq m-2$, if there is a directed edge from $u_{i-1}$ to $u_i$ (we know that $i-1$ is odd and $i$ is even), the four days of extended normal games are arranged as $R_1\cdot R_2\cdot \overline{R_1}\cdot \overline{R_2}$, otherwise, $\overline{R_3}\cdot \overline{R_4}\cdot R_3\cdot R_4$.

The design of right super-games has two advantages: First, the construction can always keep team $t_{r1}$ playing $AHHA$ and $t_{r2}$ playing $HAAH$ in each time slot, which can reduce the frequency of them returning home;
Second, we can make sure that the road trips of teams in the super-team with an index do not cause any extra cost in each time slot $q$.

\textbf{The last time slot}: Now we are ready to design the unarranged games on six days in the last time slot.
There are three kinds of unarranged games.
First, for each super-team $u_i$, by the design of left/normal/right super-games, the games between teams in $u_i$, $\{t_{2i-1}\leftrightarrow t_{2i}\}$, were not arranged.
Second, for super-teams $u_{i-1}$ and $u_i$ ($1\leq i\leq m-2$), there is a right super-game between them, and then there are two days of games unarranged.
Third, since there is no super-game between super-teams $u_l$ and $u_r$, we know that the four games between teams in $u_l$ and $u_r$ were not arranged.
Note that super-teams $(u_1,u_2,...,u_{m-2},u_1)$ form a cycle with an odd length. The unarranged games on the last six days can be seen in Figure~\ref{figb6}. We can see that each team has six unarranged games.

\begin{figure}[ht]
    \centering
    \includegraphics[scale=1]{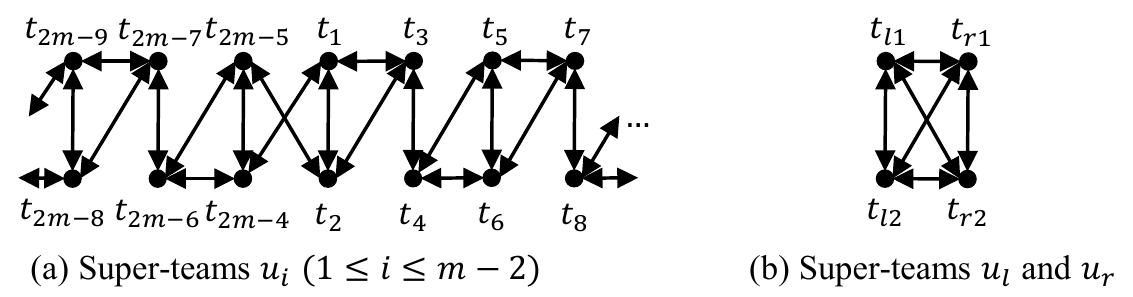}
    \caption{The left games on the last six days}
    \label{figb6}
\end{figure}

To arrange the remaining games, we introduce three days of games: $self$, $A_1$ and $A_2$. An illustration is shown in Figure~\ref{figb7}.

\begin{figure}[ht]
    \centering
    \includegraphics[scale=1]{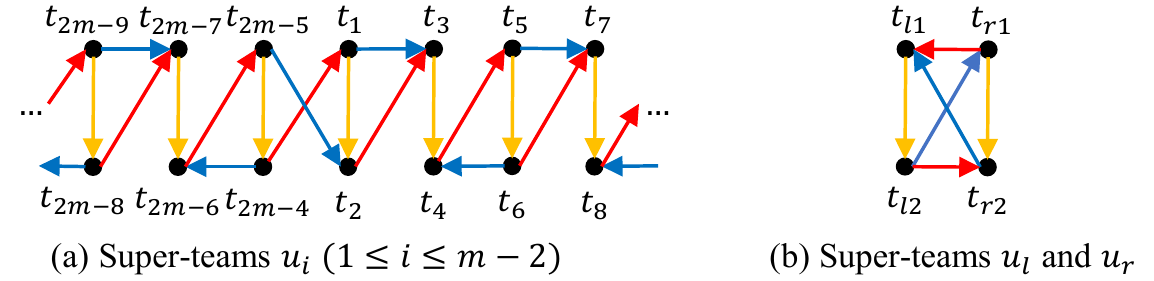}
    \caption{An illustration of the three days of games: the games in $self$ are denoted by orange edges; the games in $A_1$ are denoted by blue edges; the games in $A_2$ are denoted by red edges}
    \label{figb7}
\end{figure}

Note that we also use $\overline{self}$, $\overline{A_1}$, $\overline{A_2}$ to denote the day of games with the reversed directions in $self$, $A_1$, $A_2$, respectively, i.e., the partner is the same but the game venue changes to the other team's home. We can see that the unarranged games in Figure~\ref{figb6} can be presented by the six days of games $self \cup A_1 \cup A_2 \cup \overline{self} \cup \overline{A_1} \cup \overline{A_2}$.
Next, we arrange the six days $\{self,\overline{self},$ $A_1,A_2,$ $\overline{A_1},\overline{A_2}\}$ to combine the previous $2n-8$ days without violating the bounded-by-$k$ and no-repeat constraints. We consider two cases:
For super-teams $u_i$ ($1\leq i\leq m-2$), the six days are arranged in the order: $A_1\cdot self\cdot A_2\cdot \overline{A_1}\cdot \overline{A_2}\cdot \overline{self}$; For super-teams $u_l$ and $u_r$, the six days are arranged in the order: $self\cdot A_1\cdot A_2\cdot \overline{A_1}\cdot \overline{A_2}\cdot \overline{self}$.

We have described the main part of the scheduling algorithm. Next, we will prove its feasibility.

\begin{theorem}\label{feas2}
For TTP-$2$ with $n$ teams such that $n\geq 10$ and $n\equiv 2 \pmod 4$, the above construction can generate a feasible schedule.
\end{theorem}

\begin{proof}
First, we show that each team plays all the required $2n-2$ games in the $2n-2$ days.
According to the schedule, we can see that each team will attend one game in each of the $2n-2$ days. Furthermore,
it is not hard to observe that no two teams play against each other at the same place.
So each team will play the required $2n-2$ games.

Second, it is easy to see that each team will not violate the no-repeat property.
In any time slot, no two games between the same teams are arranged in two consecutive days. Especially, $self$ and $\overline{self}$ are not arranged on two consecutive days, which we can see in the last time slot. In two different time slots, each team will play against different teams.

Last, we prove that each team does not violate the bounded-by-$k$ property. We still use `$H$' and `$A$' to denote a home game and an away game, respectively. We will also let $\overline{H}=A$ and $\overline{A}=H$.

We first look at the games in the first $2n-12$ days. For the two teams in $u_l$, the four games in the first time slot will be $HHAA$ (see Figure~\ref{figb3}), the four games in an even time slot will be $HAAH$ (see Figure~\ref{figb4}), and the four games in an odd time slot (not containing the first time slot) will be $AHHA$. So two consecutive time slots can combine well.
For the two teams in $u_r$, the four games of team $t_{r_1}$ will always be $AHHA$, and the four games of team $t_{r_2}$ will always be $\overline{HAAH}$ in each time slot, which can be seen from the construction of right super-games.
Two consecutive time slots can still combine well.
Next, we consider a team $t_i$ in $u_j$ $(j\in \{1,2,\dots,m-2\})$.
In the time slots for away normal/right super-games (the direction of the edge is out of $u_j$), the four games will be $AAHH$, and $\overline{AAHH}$ otherwise.
In the time slots for away left super-games, the four games will be $AHHA$, and $\overline{AAHH}$ otherwise.
According to the rotation scheme of our schedule, super-team $u_j$ will always play away (resp., home) normal/right super-games until it plays an away (resp., a home) left super-games. After playing the away (resp., home) left super-games, it will always play home (resp., away) normal/right super-games. So two consecutive time slots can combine well.

Finally, we have the last ten days in the last two time slots not analyzed yet. We just list out the last ten games in the last two time slots for each team. For the sake of presentation, we also let $t_{i_1}=t_{2i-1}$ and $t_{i_2}=t_{2i}$. We will have six different cases: teams in $u_l$, $u_r$, $u_1$, $u_o$ for odd $o\in \{3,\dots, m-2\}$, and $u_e$ for even $e\in \{2,\dots, m-3\}$.
The last ten games are shown in Figure~\ref{figb8}.

\begin{figure}[ht]
    \centering
    \includegraphics[scale=0.7]{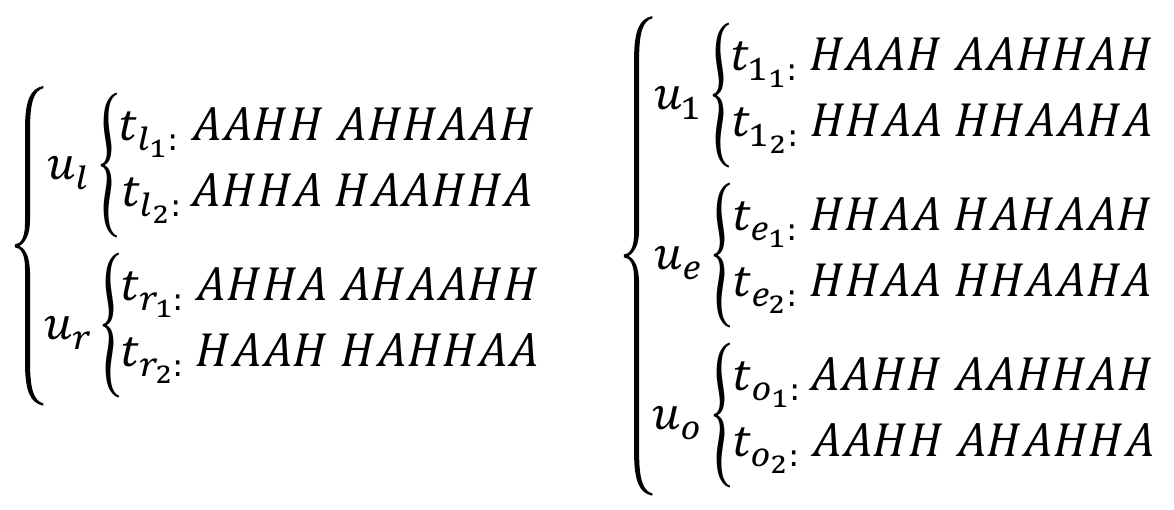}
    \caption{The last ten games for the case of $n\equiv 2\pmod 4$, where $o\in \{3,\dots, m-2\}$ and $e\in \{2,\dots, m-3\}$}
    \label{figb8}
\end{figure}

From Figure~\ref{figb8}, we can see that there are no three consecutive home/away games. It is also easy to see that on day $2n-12$ (the last day in time slot $m-3$), the games for $t_{l_1}$ and $t_{l_2}$ (in $u_l$) are $H$, the games for $t_{r_1}$ and $t_{r_2}$ (in $u_r$) are $A$ and $H$, the games for $t_{1_1}$ and $t_{1_2}$ (in $u_1$) are $A$, the games for $t_{e_1}$ and $t_{e_2}$ (in $u_e$) are $A$, and the games for $t_{o_1}$ and $t_{o_2}$ (in $u_o$) are $H$. So time slots $m-3$ and $m-2$ can also combine well without creating three consecutive home/away games.

Thus, the bounded-by-$k$ property also holds. We know our schedule is feasible for TTP-2.
\end{proof}

\subsection{Analyzing the Approximation Quality}
We still compare the itinerary of each team with its optimal itinerary.

For teams in super-teams $u_i$ ($1\leq i\leq m-1$), we can see that they stay at home before the first game in a super-game and return home after the last game in the super-game (for the last two super-games in the last two time slots, we can see it in Figure~\ref{figb8}).
We can look at the sub itinerary of a team on the four days in a left/normal super-game, which is coincident with a sub itinerary of the optimal itinerary (see the proof Lemma \ref{extra}). But, the sub itinerary of a team in a right/last super-game may be not coincident with a sub itinerary of the optimal itinerary. For example, in the right super-game between super-teams $u_1=\{t_1,t_2\}$ and $u_{m-2}=\{t_{2m-5},t_{2m-4}\}$, the sub itinerary of team $t_{1}$ contains one road trip $(t_1,t_{r1},t_{2m-5},t_1)$, while the optimal itinerary of it, containing two road trips $(t_1,t_{r1},t_{r2},t_1)$ and $(t_1,t_{2m-5},t_{2m-4},t_1)$, cannot contain a coincident sub itinerary.
So we may compare the sub itineraries in two right super-games and the last super-games together, which will be coincident with some sub itineraries in its optimal itinerary.
For teams in super-team $u_r$, since team $t_{r1}$ always plays $AHHA$ and team $t_{r2}$ always plays $HAAH$ in each right super-game of the first $m-2$ time slots, the sub itinerary of them in each of these time slots may not be coincident with a sub itinerary of the optimal itinerary. We may just compare their whole itineraries with their optimal itineraries.

In the following part, we always assume that $u_0=u_{m-2}$, i.e., $t_{-1}=t_{2m-5}$ and $t_{0}=t_{2m-4}$.
We use $\Delta_i$ to denote the sum expected extra cost of teams $t_{2i-1}$ and $t_{2i}$ in super-team $u_i$.
For the sake of presentation, we attribute the extra cost of teams in left super-games to $\Delta_{m-1}$, which will be analyzed on super-team $u_l$. Thus, when we analyze super-teams $u_i$ ($1\leq i\leq m-2$), we will not analyze the extra cost caused in its left super-game again.

\textbf{The extra cost of super-team $u_l$:}
For super-team $u_l$, there are $m-3$ left super-games (including one special left super-game) in the middle $m-3$ time slots and one last super-game in the last time slot.
As mentioned before, the design of (normal) left super-game is still the same, while the design of (special) left super-game is the same as the penultimate super-game in the case of even $n/2$.
By Lemma~\ref{extra}, we know that the expected extra cost of $m-4$ (normal) left super-games and one (special) left super-game is bounded by $(m-4)\times\frac{4}{n(n-2)}\LB+\frac{2}{n(n-2)}\LB=\frac{2n-14}{n(n-2)}\LB$.
In the last time slot, we can directly compare teams' road trips with their coincident sub itineraries of their optimal itineraries. Both teams $t_{l1}$ and $t_{l2}$ will have the same road trips as that in the optimal itineraries (see the design of the six days: $self\cdot A_1\cdot A_2\cdot \overline{A_1}\cdot \overline{A_2}\cdot \overline{self}$). We can get that
\begin{equation}\label{u_l}
\Delta_{m-1}\leq \frac{2n-14}{n(n-2)}\LB.
\end{equation}

\textbf{The extra cost of super-team $u_1$:}
For super-team $u_1$, it plays $m-5$ normal super-games, one (special) left super-game, two right super-games, and one last super-game. The normal super-games do not cause extra cost by Lemma~\ref{core}. Since the extra cost of the left super-game has been analyzed on super-team $u_l$, we only need to analyze the extra cost of the two right super-games and the last super-game. Figure~\ref{figbb1} shows their road trips in these three time slots and their coincident sub itineraries of their optimal itineraries.

\begin{figure}[ht]
    \centering
    \includegraphics[scale=0.7]{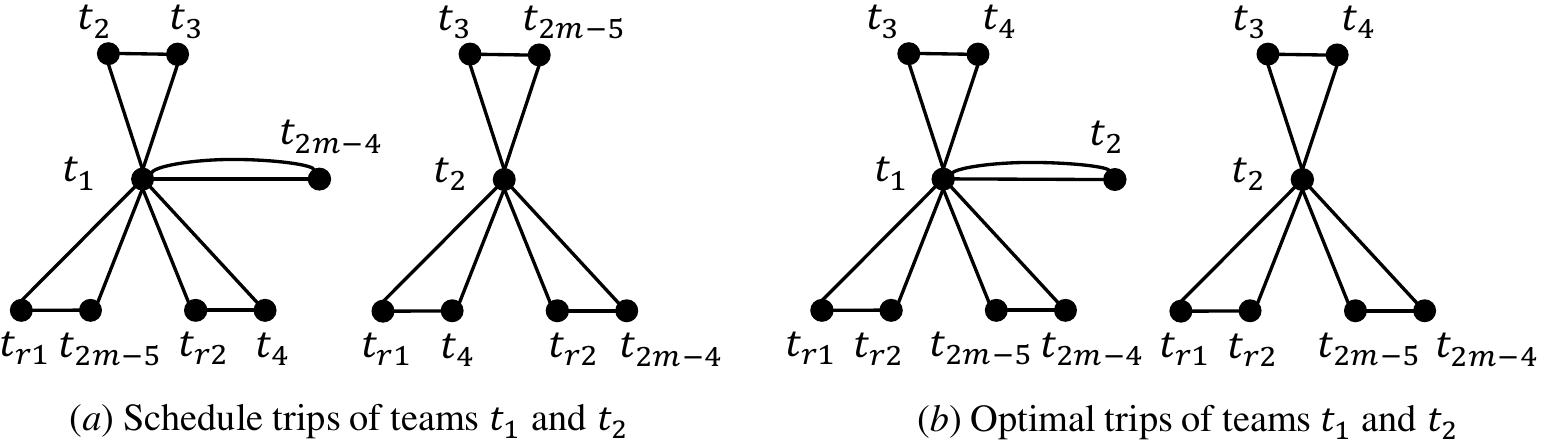}
    \caption{The road trips of teams in $u_1$ and their coincident sub itineraries of their optimal itineraries}
    \label{figbb1}
\end{figure}

By Lemma~\ref{core}, we can get that
\begin{equation}\label{u_1}
\begin{aligned}
\Delta_1 =&\ \EE{(D_{1,2m-4}+D_{2,3}+D_{r1,2m-5}+D_{r2,4})-(D_{1,2}+D_{3,4}+D_{2m-5,2m-4}+D_{r1,r2})}\\
&\ +\EE{(D_{3,2m-5}+D_{r1,4}+D_{r2,2m-4})-(D_{3,4}+D_{2m-5,2m-4}+D_{r1,r2})}\\
\leq&\ \frac{7}{n(n-2)}\LB.
\end{aligned}
\end{equation}

\textbf{The extra cost of super-team $u_i$ ($i=3,5...,m-4$):}
For super-team $u_i$, it plays $m-5$ normal super-games, one left super-game, two right super-games, and one last super-game.
Similarly, the normal super-games do not cause extra cost and the left super-game has been analyzed.
Only the two right super-games and the last super-game can cause extra cost. Figure~\ref{figbb2} shows theirs road trips in these three time slots and their coincident sub itineraries of their optimal itineraries.

\begin{figure}[ht]
    \centering
    \includegraphics[scale=0.7]{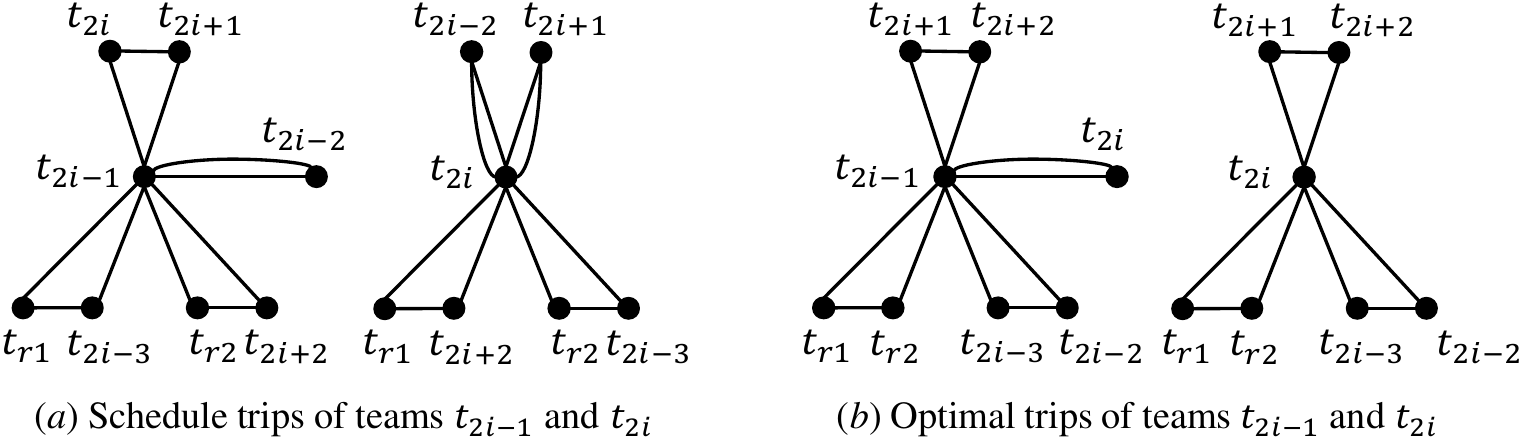}
    \caption{The road trips of teams in $u_i$ ($i=3,5...,m-4$) and their coincident sub itineraries of their optimal itineraries}
    \label{figbb2}
\end{figure}

By Lemma~\ref{core}, we can get that
\begin{equation}\label{u_oddi}
\begin{aligned}
\Delta_i =&\ \EE{(D_{2i-2,2i-1}+D_{2i,2i+1}+D_{r1,2i-3}+D_{r2,2i+2})-(D_{2i-3,2i-2}+D_{2i-1,2i}+D_{2i+1,2i+2}+D_{r1,r2})}\\
&\ +\EE{(D_{2i-2,2i}+D_{2i,2i+1}+D_{r1,2i+2}+D_{r2,2i-3})-(D_{2i-3,2i-2}+D_{2i+1,2i+2}+D_{r1,r2})}\\
\leq&\ \frac{8}{n(n-2)}\LB.
\end{aligned}
\end{equation}

\textbf{The extra cost of super-team $u_i$ ($i=2,4...,m-3$):}
For super-team $u_i$, it plays $m-5$ normal super-games, one left super-game, two right super-games, and one last super-game.
Similarly, only these two right super-games and the last super-game could cause extra cost. However, by the design of right super-games and the last super-game, the two teams in $u_i$ will have the same road trips as that in the optimal itinerary and then the extra cost is 0.

\textbf{The extra cost of super-team $u_{m-2}$:}
For super-team $u_{m-2}$, it also plays $m-5$ normal super-games, one left super-game, two right super-games, and one last super-game.
Only the two right super-games and the last super-game can cause extra cost.
Figure~\ref{figbb3} shows theirs road trips in these three time slots and their coincident sub itineraries of their optimal itineraries.

\begin{figure}[ht]
    \centering
    \includegraphics[scale=0.7]{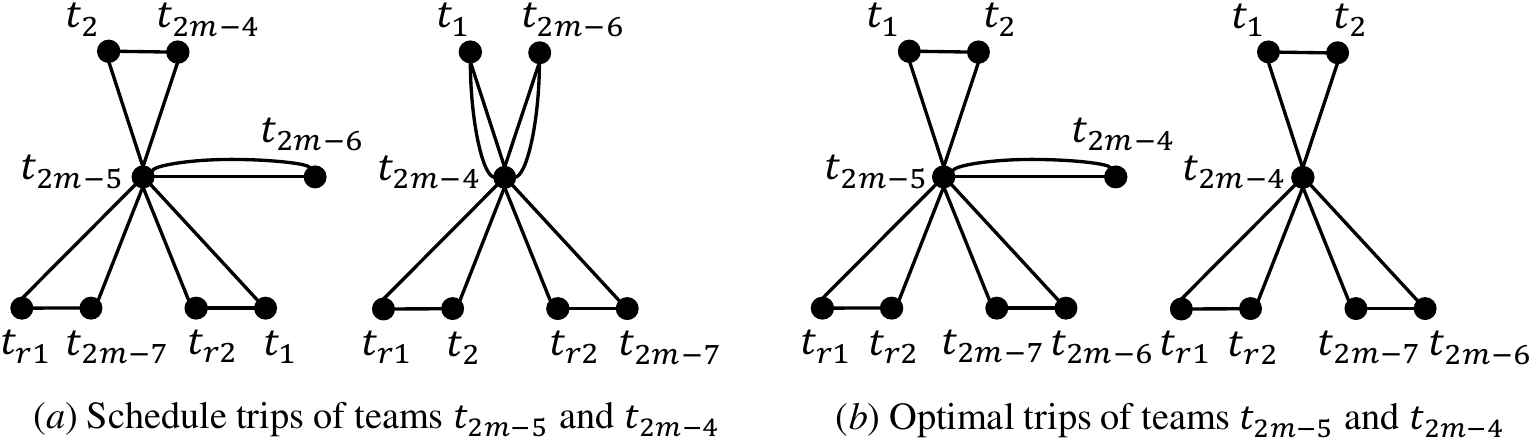}
    \caption{The road trips of teams in $u_{m-2}$ and their coincident sub itineraries of their optimal itineraries}
    \label{figbb3}
\end{figure}

By Lemma~\ref{core}, we can get that
\begin{equation}\label{u_m-2}
\begin{aligned}
\Delta_{m-2} =&\ \EE{(D_{2,2m-4}+D_{2m-6,2m-5}+D_{r1,2m-7}+D_{r2,1})-(D_{1,2}+D_{2m-7,2m-6}+D_{2m-5,2m-4}+D_{r1,r2})}\\
&\ +\EE{(D_{1,2m-4}+D_{2m-6,2m-4}+D_{r1,2}+D_{r2,2m-7})-(D_{1,2}+D_{2m-7,2m-6}+D_{r1,r2})}\\
\leq&\ \frac{8}{n(n-2)}\LB.
\end{aligned}
\end{equation}

By (\ref{u_l}), (\ref{u_1}), (\ref{u_oddi}), and (\ref{u_m-2}), we can get that

\begin{equation}\label{sum1}
\sum_{i=1}^{m-2}\Delta_i\leq \frac{7}{n(n-2)}\LB+\frac{m-3}{2}\times\frac{8}{n(n-2)}\LB+\frac{2n-14}{n(n-2)}\LB=\frac{4n-19}{n(n-2)}\LB.
\end{equation}

\textbf{The extra cost of super-team $u_r$:}
For super-team $u_r$, there are two teams $t_{r1}$ and $t_{r2}$.
For team $t_{r1}$, there is one road trip visiting $t_{l1}$ and $t_{l2}$ in the last time slot (see the design of the six days: $self\cdot A_1\cdot A_2\cdot \overline{A_1}\cdot \overline{A_2}\cdot \overline{self}$), which does not cause extra cost. So we will not compare this sub itinerary. Figure~\ref{figbb4} shows the other road trips of $t_{r1}$ in the first $m-2$ time slots and the whole road trips of $t_{r2}$ in $m-1$ time slots. Note that the road trips in Figure~\ref{figbb4} correspond to the case of $n\equiv 2 \pmod 8$.

\begin{figure}[ht]
    \centering
    \includegraphics[scale=0.7]{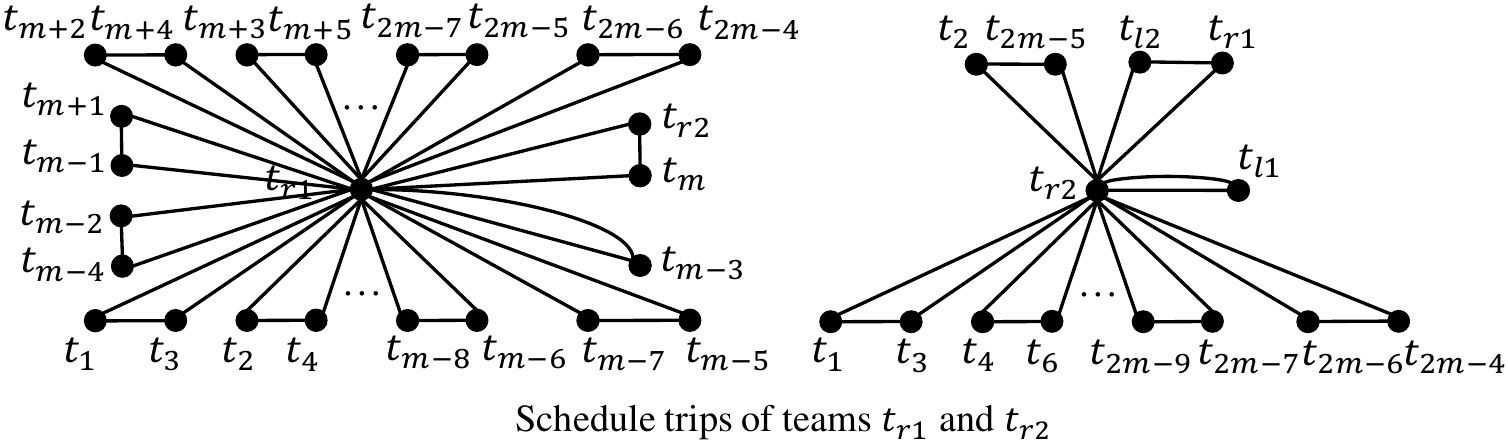}
    \caption{The road trips of teams in $u_{r}$ for the case of $n\equiv 2 \pmod 8$}
    \label{figbb4}
\end{figure}

By Lemma~\ref{core}, for the case of $n\equiv 2\pmod 8$, we can get that
\begin{equation}\label{u_r}
\begin{aligned}
\Delta_m =&\ \EEE{\sum_{i=1}^{\frac{m-5}{4}}(D_{4i-3,4i-1}+D_{4i-2,4i})+\sum_{i=\frac{m+3}{4}}^{\frac{m-3}{2}}(D_{4i-1,4i+1}+D_{4i,4i+2})}\\
&\ +\EEE{D_{m-4,m-2}+D_{m-1,m+1}+D_{r1,m-3}+D_{r2,m}-\sum_{i=1}^{m-2}D_{2i-1,2i}-D_{r1,r2}}\\
&\ +\EEE{\sum_{i=1}^{\frac{m-3}{2}}(D_{4i-3,4i-1}+D_{4i,4i+2})+(D_{2,2m-5}+D_{r1,l2}+D_{r2,l1})-\sum_{i=1}^{m}D_{2i-1,2i}}\\
\leq&\ \lrA{2\times\frac{m-5}{4}+2\times\frac{m-5}{4}+4+2\times\frac{m-3}{2}+3}\times\frac{1}{n(n-2)}\LB\\
=&\ \frac{n-1}{n(n-2)}\LB.
\end{aligned}
\end{equation}

For the case of $n\equiv 6 \pmod 8$, the dash edges will be reversed (see Figure \ref{figb1}) and then the road trips of teams in $u_r$ may be slightly different. Figure~\ref{figbb5} shows the other road trips of $t_{r1}$ in the first $m-2$ time slots and the whole road trips of $t_{r2}$ in $m-1$ time slots. We can see that the road trips of team $t_{r2}$ are still the same while the road trips of team $t_{r1}$ are different.

\begin{figure}[ht]
    \centering
    \includegraphics[scale=0.7]{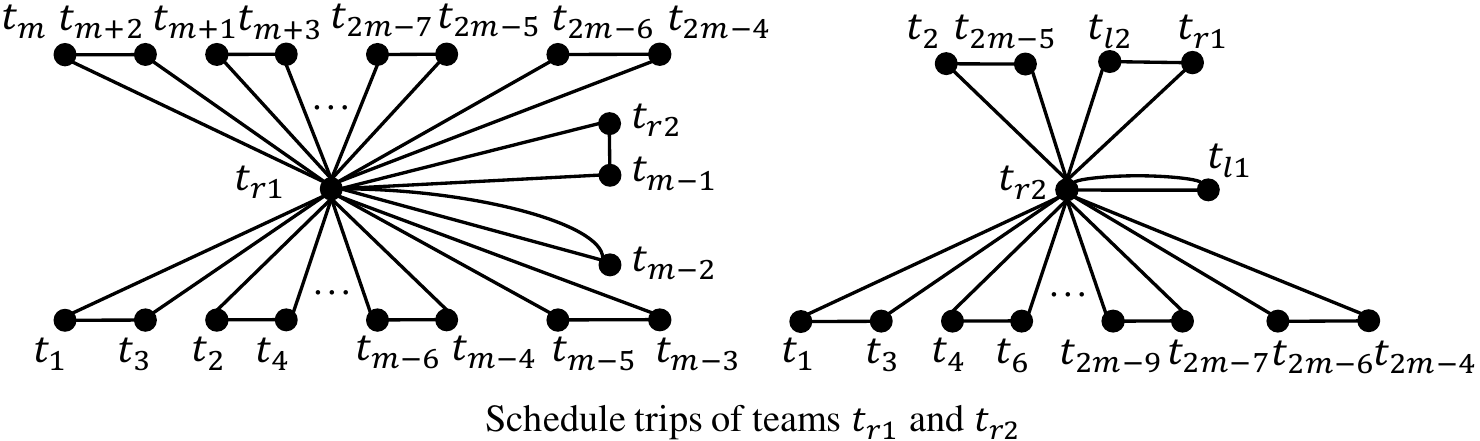}
    \caption{The road trips of teams in $u_{r}$ for the case of $n\equiv 6 \pmod 8$}
    \label{figbb5}
\end{figure}

By Lemma~\ref{core}, for the case of $n\equiv 6\pmod 8$, we can get that
\begin{equation}\label{u_r+}
\begin{aligned}
\Delta_m =&\ \EEE{\sum_{i=1}^{\frac{m-3}{4}}(D_{4i-3,4i-1}+D_{4i-2,4i})+\sum_{i=\frac{m+1}{4}}^{\frac{m-3}{2}}(D_{4i-1,4i+1}+D_{4i,4i+2})}\\
&\ +\EEE{D_{r1,m-2}+D_{r2,m-1}-\sum_{i=1}^{m-2}D_{2i-1,2i}-D_{r1,r2}}\\
&\ +\EEE{\sum_{i=1}^{\frac{m-3}{2}}(D_{4i-3,4i-1}+D_{4i,4i+2})+D_{2,2m-5}+D_{r1,l2}+D_{r2,l1}-\sum_{i=1}^{m}D_{2i-1,2i}}\\
\leq&\ \lrA{2\times\frac{m-3}{4}+2\times\frac{m-3}{4}+2+2\times\frac{m-3}{2}+3}\times\frac{1}{n(n-2)}\LB\\
=&\ \frac{n-1}{n(n-2)}\LB.
\end{aligned}
\end{equation}
Therefore, the upper bound of $\Delta_m$ in these two cases is identical.

By (\ref{sum1}), (\ref{u_r}), and (\ref{u_r+}), we can get that the
total expected extra cost is
\[
\sum_{i=1}^{m}\Delta_i\leq\frac{4n-19}{n(n-2)}\LB+\frac{n-1}{n(n-2)}\LB=\lrA{\frac{5}{n}-\frac{10}{n(n-2)}}\LB.
\]

\begin{theorem}\label{res_2}
For TTP-$2$ with $n$ teams, when $n\geq 10$ and $n\equiv 2 \pmod 4$, there is a randomized $O(n^3)$-time algorithm with an expected approximation ratio of $1+{\frac{5}{n}}-{\frac{10}{n(n-2)}}$.
\end{theorem}

The analysis of our algorithm for odd $n/2$ is also tight.
We can also consider the example, where the distance of each edge in the minimum matching $M$ is 0 and the distance of each edge in $G-M$ is 1. Recall that the independent lower bound satisfies $\LB=n(n-2)$.
For odd $n/2$, by the analysis of $\Delta_i$, we can compute that the total extra cost of our construction is $5n-20$. Thus, in this case, the ratio is
\[
1+\frac{5n-20}{n(n-2)}=1+\frac{5}{n}-\frac{10}{n(n-2)}.
\]

\section{The Derandomization}
In the previous sections, we proposed a randomized $(1+\frac{4L(n)+n}{n(n-2)})$-approximation algorithm for even $n/2$ and a randomized $(1+\frac{5}{n}-\frac{10}{n(n-2)})$-approximation algorithm for odd $n/2$. In this section, we show how to derandomize our algorithms efficiently by using the method of conditional expectations~\cite{motwani1995randomized}.
This method was also used to derandomize a $(2+O(1/n))$-approximation algorithm for TTP-3 in \cite{miyashiro2012approximation}.
Their algorithm randomly orders all teams while our algorithms first randomly orders super-teams and then randomly orders the teams in each super-team. This is the difference.
We will also analyze a running-time bound of the derandomization.

According to the analysis of our algorithms, the total extra cost is bounded by
\begin{equation}\label{de1}
W=\sum_{\substack{1\leq i'<j'\leq m\\t_i\in u_{i'} \& t_j\in u_{j'}}} n_{ij}D_{i,j},
\end{equation}
where $n_{ij}$ is the number of times edge $t_it_j$ appears in computing the total extra cost.

In the main framework of our algorithms, there are two steps using the randomized methods: the first is that we use $\{u_1,\dots,u_m\}$ to label the $m$ edges in $M$ randomly; the second is that we use $\{t_{2i-1}, t_{2i}\}$ to label the two teams in each super-team $u_i$ randomly.

We first consider the derandomization of super-teams, i.e., use $\{u_1,\dots,u_m\}$ to label the $m$ edges in $M$ in a deterministic way such that the expected approximation ratio still holds.

\subsection{The Derandomization of Super-teams}
Suppose the $m$ edges in $M$ are denoted by $\{e_1,\dots, e_m\}$. We wish to find a permutation $\sigma: (1,2,\dots,m)\leftrightarrow (\sigma_1, \sigma_2,\dots,\sigma_m)$ such that
\[
\EE{W|u_1=e_{\sigma_1},u_2=e_{\sigma_2},\dots,u_m=e_{\sigma_m}}\leq \EE{W}.
\]
We can determine each $\sigma_i$ sequentially. Suppose we have already determine $(\sigma_1, \sigma_2,\dots,\sigma_{s-1})$ such that
\[
\EE{W|u_1=e_{\sigma_1},u_2=e_{\sigma_2},\dots,u_{s-1}=e_{\sigma_{s-1}}}\leq \EE{W}.
\]
To determine $\sigma_{s}$, we can simply let $\sigma_s$ be
\begin{equation}\label{de2}
\sigma_s=\arg\min_{\sigma_s\in \{1,2,\dots,m\}\setminus\{\sigma_1, \sigma_2,\dots,\sigma_{s-1}\}} \EE{W|u_1=e_{\sigma_1},u_2=e_{\sigma_2},\dots,u_s=e_{\sigma_s}}\leq \EE{W}.
\end{equation}
Then, we can get
\begin{equation}\label{de3}
\EE{W|u_1=e_{\sigma_1},u_2=e_{\sigma_2},\dots,u_s=e_{\sigma_s}}\leq \EE{W|u_1=e_{\sigma_1},u_2=e_{\sigma_2},\dots,u_{s-1}=e_{\sigma_{s-1}}}\leq \EE{W}.
\end{equation}
Therefore, we can repeat this procedure to determine the permutation $\sigma$.

Next, we show how to compute $\EE{W|u_1=e_{\sigma_1},u_2=e_{\sigma_2},\dots,u_s=e_{\sigma_s}}$. Recall that $D(u_{i'},u_{j'})=\sum_{t_{i}\in u_{i'} \& t_{j}\in u_{j'}}D_{i,j}$.
When $t_i\in u_{i'}$ and $t_j\in u_{j'}$ ($i'\neq j'$), we can get that
\[
\EE{D_{i,j}|u_1=e_{\sigma_1},u_2=e_{\sigma_2},\dots,u_s=e_{\sigma_s}}=\frac{1}{4}\EE{D(u_{i'},u_{j'})|u_1=e_{\sigma_1},u_2=e_{\sigma_2},\dots,u_s=e_{\sigma_s}},
\]
since the two teams in each super-team are still labeled randomly.
Hence, by (\ref{de1}) and (\ref{de3}), we can get
\[
\begin{aligned}
&\EE{W|u_1=e_{\sigma_1},u_2=e_{\sigma_2},\dots,u_s=e_{\sigma_s}}\\
&=\ \EEE{\sum_{\substack{1\leq i'<j'\leq m\\t_i\in u_{i'} \& t_j\in u_{j'}}} n_{ij}D_{i,j}|u_1=e_{\sigma_1},u_2=e_{\sigma_2},\dots,u_s=e_{\sigma_s}}\\
&=\ \sum_{\substack{1\leq i'<j'\leq m\\t_i\in u_{i'} \& t_j\in u_{j'}}}n_{ij}\EE{ D_{i,j}|u_1=e_{\sigma_1},u_2=e_{\sigma_2},\dots,u_s=e_{\sigma_s}}\\
&=\ \sum_{\substack{1\leq i'<j'\leq m\\t_i\in u_{i'} \& t_j\in u_{j'}}}\frac{1}{4}n_{ij}\EE{D(u_{i'},u_{j'})|u_1=e_{\sigma_1},u_2=e_{\sigma_2},\dots,u_s=e_{\sigma_s}}\\
&=\ \sum_{\substack{1\leq i'<j'\leq m}}\frac{1}{4}m_{ij}\EE{D(u_{i'},u_{j'})|u_1=e_{\sigma_1},u_2=e_{\sigma_2},\dots,u_s=e_{\sigma_s}},
\end{aligned}
\]
where $m_{i'j'}=\sum_{t_i\in u_{i'} \& t_j\in u_{j'}}n_{ij}$.
Let $S_s=\{\sigma_1, \sigma_2,\dots,\sigma_s\}$, $\overline{S_s}=\{1,2,\dots,m\}\setminus S_s$, $T_s=\{1, 2,\dots,s\}$, and $\overline{T_s}=\{1,2,\dots,m\}\setminus T_s$.
The value $\EE{D(u_{i'},u_{j'})|u_1=e_{\sigma_1},u_2=e_{\sigma_2},\dots,u_s=e_{\sigma_s}}$ can be computed as follows:
\[
\begin{aligned}
&\EE{D(u_{i'},u_{j'})|u_1=e_{\sigma_1},u_2=e_{\sigma_2},\dots,u_s=e_{\sigma_s}} &\\
&=\left\{
\begin{array}{*{20}l}
D(u_{i'},u_{j'}), & i'\in T_s, j'\in T_s,\\
\frac{1}{m-s}\sum_{\sigma_{j'}\in\overline{S_s}}D(u_{i'}, e_{\sigma_{j'}}), & i'\in T_s, j'\in \overline{T_s},\\
\frac{1}{m-s}\sum_{\sigma_{i'}\in\overline{S_s}}D(e_{\sigma_{i'}}, u_{j'}), & i'\in \overline{T_s}, j'\in T_s,\\
\frac{1}{(m-s-1)(m-s)}\sum_{\sigma_{i'},\sigma_{j'}\in\overline{T_s}}D(e_{\sigma_{i'}}, e_{\sigma_{j'}}), & i'\in \overline{T_s}, j'\in \overline{T_s},\\
\end{array}
\right.
\end{aligned}
\]
where $D(e_{\sigma_{i'}}, e_{\sigma_{j'}})$ is the sum distance of all four edges between vertices of $e_{\sigma_{i'}}$ and vertices of $e_{\sigma_{j'}}$ (the edge $e_{\sigma_{i'}}$/$e_{\sigma_{j'}}$ can be regarded as a super-team).

Next, we analyze the running time of our derandomization on super-teams.
When $i'\in T_s$ and $j'\in \overline{T_s}$, there are $O(s^2)=O(n^2)$ variables, the expected conditional value of each variable $D(u_{i'},u_{j'})$ can be computed in $O(1)$ time, and hence the expected conditional values of these variables can be computed in $O(n^2)$ time.
When $i'\in T_s$ and $j'\in \overline{T_s}$, there are $O(s(m-s))=O(n^2)$ variables, the expected conditional value of each variable $D(u_{i'},u_{j'})$, which is not related to $j'$, can be computed in $O(m-s)=O(n)$ time, and hence these variables can be computed in $O(ns)=O(n^2)$ time.
Similarly, when $i'\in \overline{T_s}$ and $j'\in T_s$, these variables can be computed in $O(n^2)$ time.
When $i'\in \overline{T_s}$ and $j'\in \overline{T_s}$, there are $O((m-s)^2)=O(n^2)$ variables, the expected conditional value of each variable $D(u_{i'},u_{j'})$, which is not related to $i'$ and $j'$ (it is a constant), can be computed in $O(s^2)=O(n^2)$ time, and hence these variables can be computed in $O(ns)=O(n^2)$ time.
Therefore, all of $O(n^2)$ variables can be computed in $O(n^2)$ time. Hence, the value $\EE{W|u_1=e_{\sigma_1},u_2=e_{\sigma_2},\dots,u_s=e_{\sigma_s}}$ can be computed in $O(n^2)$ time. To determine $\sigma_{s}$, by (\ref{de2}), we need to use $O(n^3)$ time. Therefore, to determine the permutation $\sigma$, the total running time is $O(n^4)$.

Next, we consider the derandomization of each super-team, i.e., use $\{t_{2i-1}, t_{2i}\}$ to label the two vertices of edge $e_{\sigma_i}$ in a deterministic way such that the expected approximation ratio still holds.

\subsection{The Derandomization of Each Super-team}
Now, we assume that the $m$ edges in $M$ are directed edges. Let $a(e)$ and $b(e)$ be the tail vertex and the head vertex of the directed edge $e$, respectively. In the previous derandomization, we have determined $u_i=e_{\sigma_i}$, i.e., the super-team $u_i$ refers to the edge $e_{\sigma_i}$. Hence, the weight in (\ref{de1}) can be written as $W(\sigma)$. Note that $\EE{W(\sigma)}\leq \EE{W}$. Then, we need to determine the labels of $a(e_{\sigma_i})$ and $b(e_{\sigma_i})$. We use $u_i=e^{0}_{\sigma_i}$ (resp., $u_i=e^{1}_{\sigma_i}$) to mean that we let $t_{2i-1}=a(e_{\sigma_i})$ and $t_{2i}=b(e_{\sigma_i})$ (resp., $t_{2i-1}=b(e_{\sigma_i})$ and $t_{2i}=a(e_{\sigma_i})$).
We wish to find a vector $(\pi_1,\pi_2,\dots,\pi_m)\in\{0,1\}^m$ such that
\[
\EE{W(\sigma)|u_1=e^{\pi_1}_{\sigma_1},u_2=e^{\pi_2}_{\sigma_2},\dots,u_m=e^{\pi_m}_{\sigma_m}}\leq \EE{W(\sigma)}.
\]
The idea of derandomization is the same. We can determine each $\pi_i$ sequentially.
Suppose we have already determine $(\pi_1, \pi_2,\dots,\pi_{s-1})$ such that
\[
\EE{W(\sigma)|u_1=e^{\pi_1}_{\sigma_1},u_2=e^{\pi_2}_{\sigma_2},\dots,u_{s-1}=e^{\pi_{s-1}}_{\sigma_{s-1}}}\leq \EE{W(\sigma)}.
\]
To determine $\pi_{s}$, using a similar argument, we know that we can let $\pi_s$ be
\begin{equation}\label{de22}
\pi_s=\arg\min_{\pi_s\in \{0,1\}} \EE{W(\sigma)|u_1=e^{\pi_1}_{\sigma_1},u_2=e^{\pi_2}_{\sigma_2},\dots,u_s=e^{\pi_s}_{\sigma_s}}\leq \EE{W(\sigma)}.
\end{equation}

Next, we show how to compute $\EE{W(\sigma)|u_1=e^{\pi_1}_{\sigma_1},u_2=e^{\pi_2}_{\sigma_2},\dots,u_s=e^{\pi_s}_{\sigma_s}}$.
By (\ref{de1}), we can get
\[
\begin{aligned}
&\EE{W(\sigma)|u_1=e^{\pi_1}_{\sigma_1},u_2=e^{\pi_2}_{\sigma_2},\dots,u_s=e^{\pi_s}_{\sigma_s}}\\
&=\ \EEE{\sum_{\substack{1\leq i'<j'\leq m\\t_i\in u_{i'} \& t_j\in u_{j'}}} n_{ij}D_{i,j}|u_1=e^{\pi_1}_{\sigma_1},u_2=e^{\pi_2}_{\sigma_2},\dots,u_s=e^{\pi_s}_{\sigma_s}}\\
&=\ \sum_{\substack{1\leq i'<j'\leq m\\t_i\in u_{i'} \& t_j\in u_{j'}}}n_{ij}\EE{ D_{i,j}|u_1=e^{\pi_1}_{\sigma_1},u_2=e^{\pi_2}_{\sigma_2},\dots,u_s=e^{\pi_s}_{\sigma_s}}.
\end{aligned}
\]
Note that $t_i\in u_{i'}$ and $t_j\in u_{j'}$.
Let $T_s=\{1, 2,\dots,s\}$ and $\overline{T_s}=\{1,2,\dots,m\}\setminus T_s$.
The value $\EE{D_{i,j}|u_1=e^{\pi_1}_{\sigma_1},u_2=e^{\pi_2}_{\sigma_2},\dots,u_s=e^{\pi_s}_{\sigma_s}}$ can be computed as follows:
\[
\begin{aligned}
&\EE{D_{i,j}|u_1=e^{\pi_1}_{\sigma_1},u_2=e^{\pi_2}_{\sigma_2},\dots,u_s=e^{\pi_s}_{\sigma_s}} &\\
&=\left\{
\begin{array}{*{20}l}
D_{i,j}, & i'\in T_s, j'\in T_s,\\
D(t_i, u_{j'})/2, & i'\in T_s, j'\in \overline{T_s},\\
D(u_{i'}, t_j)/2, & i'\in \overline{T_s}, j'\in T_s,\\
D(u_{i'}, u_{j'})/4, & i'\in \overline{T_s}, j'\in \overline{T_s},
\end{array}
\right.
\end{aligned}
\]
where $D(t_i, u_{j'})=\sum_{t_j\in u_{j'}}D_{i,j}$ and $D(u_{i'}, t_j)=\sum_{t_i\in u_{i'}}D_{i,j}$.

Using a similar argument, the value $\EE{D_{i,j}|u_1=e^{\pi_1}_{\sigma_1},u_2=e^{\pi_2}_{\sigma_2},\dots,u_s=e^{\pi_s}_{\sigma_s}}$ can be computed in $O(n^2)$ time. To determine $\pi_{s}$, by (\ref{de22}), we need to take $O(n^2)$ time. Therefore, to determine the vector $(\pi_1,\pi_2,\dots,\pi_m)$, the total running time is $O(n^3)$.

Therefore, the derandomization takes $O(n^4)$ extra time in total.

\section{Experimental Results}
To evaluate the performance of our schedule algorithms, we implement our algorithms and test them on well-known benchmark instances.
In fact, the full derandomizations in our algorithms are not efficient in practical.
We will use some heuristic methods to get a good order of teams in our schedules. Thus, in the implementation we use the randomized algorithms and combine them with some simple local search heuristics.

\subsection{The Main Steps in Experiments}
In our algorithms and experiments, we first compute a minimum weight perfect matching.
After that, we pack each edge in the matching as a super-team and randomly order them with $\{u_1,u_2,...,u_{m}\}$.
For each super-team $u_i$, we also randomly order the two teams in it with $\{t_{2i-1}, t_{2i}\}$. Using the obtained order, we can get a feasible solution according to our construction methods. To get possible improvements, we will use two simple swapping rules: the first is to swap two super-teams, and the second is to swap the two teams in each super-team. These two swapping rules can keep the pairs of teams in super-teams corresponding to the minimum weight perfect matching. The details of the two swapping rules are as follows.

\medskip
\noindent
\textbf{The first swapping rule on super-teams:} Given an initial schedule obtained by our randomized algorithms, where
the super-teams $\{u_1,u_2,...,u_{m}\}$ are randomly ordered. We are going to swap the positions of some pairs of super-teams $(u_i, u_j)$ to reduce the total traveling distance.
There are $\frac{m(m-1)}{2}$ pairs of super-teams $(u_i, u_j)$ with $1\leq i<j\leq m$.
We consider the $\frac{m(m-1)}{2}$ pairs $(i,j)$ in an order by dictionary.
From the first pair to the last pair $(i,j)$, we test whether the total traveling distance can be reduced after we swap the positions of the two super-teams $u_i$ and $u_j$. If no, we do not swap them and go to the next pair. If yes, we swap the two super-teams and go to the next pair. After considering all the $\frac{m(m-1)}{2}$ pairs, if there is an improvement, we repeat the whole procedure. Otherwise, the procedure ends.

\medskip
\noindent
\textbf{The second swapping rule on teams in each super-team:}
There are $m$ super-teams $u_i$ in our algorithm. For each super-team $u_i$, there are two normal teams in it, which are randomly ordered initially. We are going to swap the positions of two normal teams in a super-team to reduce the total traveling distance.
We consider the $m$ super-teams $u_i$ in an order. For each super-team $u_i$, we test whether the total traveling distance can be reduced after we swap the positions of the two teams in the supper-team $u_i$. If no, we do not swap them and go to the next super-team. If yes, we swap the two teams and go to the next super-team. After considering all the $m$ super-teams, if there is an improvement, we repeat the whole procedure. Otherwise, the procedure ends.

\medskip
In our experiments, one suit of swapping operations is to iteratively apply the two swapping rules until no improvement we can further get.
Since the initial order of the teams is random, we may generate several different initial orders and apply the swapping operations on each of them.
In our experiments, when we say excusing $x$ rounds, it means that we generate $x$ initial random orders, apply the suit of swapping operations on each of them, and return the best result.

\subsection{Applications to Benchmark Sets}
We implement our schedule algorithms to solve the benchmark instances in~\cite{trick2007challenge}.
The website introduces 62 instances, most of which were reported from real-world sports scheduling scenarios, such as the Super 14 Rugby League, the National Football League, and the 2003 Brazilian soccer championship. The number of teams in the instances varies from 4 to 40.
There are 34 instances of even $n/2$, and 28 instances of odd $n/2$.
Almost half of the instances are very small ($n\leq 8$) or very special (all teams are on a cycle or the distance between any two teams is 1), and they were not tested in previous papers~\cite{thielen2012approximation,DBLP:conf/mfcs/XiaoK16}. So we only test the remaining 33 instances, including 17 instances of even $n/2$ and 16 instances of odd $n/2$.
Due to the difference between the two algorithms for even $n/2$ and odd $n/2$, we will show the results separately for these two cases.

\medskip
\noindent
\textbf{Results of even $n/2$:}
For the case of even $n/2$, we compare our results with the best-known results in Table~\ref{experimentresult1}. In the table, the column
`\emph{Lower Bounds}' indicates the independent lower bounds;
`\emph{Previous Results}' lists previous known results in~\cite{DBLP:conf/mfcs/XiaoK16};
`\emph{Initial Results}' shows the results given by our initial randomized schedule;
`\emph{$x$ Rounds}' shows the best results after generating $x$ initial randomized schedules and applying the suit of swapping operations on them (we show the results for $x=1,10, 50$ and 300);
`\emph{Our Gap}' is defined to be $\frac{300~Rounds~-~Lower~Bounds}{Lower~Bounds}$, and `\emph{Improvement Ratio}' is defined as $\frac{Previous~Results~-~300~Rounds}{Previous~Results}$.

\begin{table}[ht]
\footnotesize
\centering
\begin{tabular}{c|ccccccccc}
\hline
Data & Lower  & Previous & Initial & 1  & 10 & 50  & 300 & Our & Improvement \\
Set & Bounds &  Results & Results & Round & Rounds  & Rounds & Rounds & Gap & Ratio\\
\hline
Galaxy40 & 298484 & 307469 &314114& 305051 & 304710 & 304575 & ${304509}$ & 2.02 & 0.96\\
Galaxy36 & 205280 & 212821 &218724& 210726 & 210642 & 210582 & ${210461}$ & 2.52 & 1.11
\\
Galaxy32 & 139922 & 145445 &144785& 142902 & 142902 & 142834 & ${142781}$ & 2.04 & 1.83
\\
Galaxy28 & 89242 & 93235 &94173& 92435 & 92121 & 92105 & ${92092}$ & 3.19 & 1.23
\\
Galaxy24 & 53282 & 55883 &55979& 54993 & ${54910}$ & 54910 & 54910 & 3.06 & 1.74\\
Galaxy20 & 30508 & 32530 &32834& 32000 & 31926 & ${31897}$ & 31897 & 4.55 & 1.95
\\
Galaxy16 & 17562 & 19040 &18664& 18409 & ${18234}$ & 18234 & 18234 & 3.83 & 4.23
\\
Galaxy12 & 8374 & 9490 &9277& 8956 & 8958 & ${8937}$ & 8937 & 6.72 & 5.83\\
NFL32 & 1162798 & 1211239 &1217448& 1190687 & 1185291 & 1185291 & ${1184791}$ & 1.89 & 2.18\\
NFL28 & 771442 & 810310 &818025& 800801 & 796568 & ${795215}$ & 795215 & 3.08 & 1.86
\\
NFL24 & 573618 & 611441 &602858& 592422 & 592152 & ${591991}$ & 591991 & 3.20 & 3.18
\\
NFL20 & 423958 & 456563 &454196& 443718 & ${441165}$ & 441165 & 441165 & 4.06 & 3.37
\\
NFL16 & 294866 & 321357 &312756& ${305926}$ & 305926 & 305926 & 305926 & 3.75 & 4.80
\\
NL16 & 334940 & 359720 &355486& 351250 & ${346212}$ & 346212 & 346212 & 3.37 & 3.76
\\
NL12 & 132720 & 144744 &146072& 139394 & ${139316}$ & 139316 & 139316 & 4.97 & 3.75\\
Super12 & 551580 & 612583 &613999& ${586538}$ & 586538 & 586538 & 586538 & 6.34 & 4.25
\\
Brazil24 & 620574 & 655235 &668236& 642251 & ${638006}$ & 638006 & 638006 & 2.81 & 2.63\\
\end{tabular}
\caption{Experimental results for even $n/2$  with an average improvement of 2.86\%}
\label{experimentresult1}
\end{table}

\medskip
\noindent
\textbf{Results of odd $n/2$:}
For odd $n/2$, we compare our results with the best-known results in Table~\ref{experimentresult2}.
Note now the previous known results in column `\emph{Previous Results}' is from another reference~\cite{thielen2012approximation}.

\begin{table}[ht]
\footnotesize
\centering
\begin{tabular}{c|ccccccccc}
\hline
Data & Lower  & Previous &Initial& 1 & 10 & 50  & 300 & Our & Improvement\\
Set & Bounds &  Results &Results& Round & Rounds  & Rounds & Rounds & Gap & Ratio\\
\hline
Galaxy38 & 244848 & 274672 &268545& 256430 & 255678 & ${255128}$ & 255128 & 4.20 & 7.12\\
Galaxy34 & 173312 & 192317 &188114& 180977 & 180896 & ${180665}$ & 180665 & 4.24 & 6.06\\
Galaxy30 & 113818 & 124011 &123841& 119524 & 119339 & 119122 & ${119076}$ & 4.62 & 3.98\\
Galaxy26 & 68826 & 77082 &75231& 73108 & 72944 & 72693 & ${72639}$ & 5.54 & 5.76\\
Galaxy22 & 40528 & 46451 &45156& 43681 & 43545 & 43478 & ${43389}$ & 7.06 & 6.59\\
Galaxy18 & 23774 & 27967 &27436& 26189 & ${26020}$ & 26020 & 26020 & 9.45 & 6.96\\
Galaxy14 & 12950 & 15642 &15070& 14540 & 14507 & ${14465}$ & 14465 & 11.70 & 7.52\\
Galaxy10 & 5280 & 6579 &6153& 5917 & ${5915}$ & 5915 & 5915 & 12.03 & 10.09
\\
NFL30 & 951608 & 1081969 &1051675& 1008487 & 1005000 & 1002665 & ${1001245}$ & 5.22 & 7.46\\
NFL26 & 669782 & 779895 &742356& 715563 & 715563 & ${714675}$ & 714675 & 6.70 & 8.36\\
NFL22 & 504512 & 600822 &584624& 548702 & 545791 & ${545142}$ & 545142 & 8.05 & 9.27\\
NFL18 & 361204 & 439152 &418148& 400390 & 398140 & ${397539}$ & 397539 & 10.06 & 9.48\\
NL14 & 238796 & 296403 &281854& 269959 & ${266746}$ & 266746 & 266746 & 11.70 & 10.01\\
NL10 & 70866 & 90254 &83500& 81107 & 80471 & ${80435}$ & 80435 & 13.50 & 10.88\\
Super14 & 823778 & 1087749 &1025167& ${920925}$ & 920925 & 920925 & 920925 & 11.79 & 15.34\\
Super10 & 392774 & 579862 &529668& 503275 & ${500664}$ & 500664 & 500664 & 27.47 & 13.66\\
\end{tabular}
\caption{Experimental results for odd $n/2$ with an average improvement of 8.65\%}
\label{experimentresult2}
\end{table}

From Tables~\ref{experimentresult1} and \ref{experimentresult2}, we can see that our algorithm improves all the 17 instances of even $n/2$ and the 16 instances of odd $n/2$.
For one round, the average improvement on the 17 instances of even $n/2$ is $2.49\%$, and the the average improvement on the 16 instances of odd $n/2$ is $8.18\%$.
For 10 rounds, the improvements will be $2.81\%$ and $8.51\%$, respectively. For 50 rounds, the improvements will be $2.85\%$ and $8.63\%$, respectively.
If we run more rounds, the improvement will be very limited and it is almost no improvement after 300 rounds.

Another important issue is about the running time of the algorithms. Indeed, our algorithms are very effective.
Our algorithms are coded using C-Free 5.0, on a standard desktop computer with a 3.20GHz AMD Athlon 200GE CPU and 8 GB RAM.
Under the above setting, for one round, all the 33 instances can be solved together within 1 second. If we run $x$ rounds, the running time will be less than $x$ seconds.
Considering that algorithms are already very fast, we did not specifically optimize the codes.
The codes of our algorithms can be found in \url{https://github.com/JingyangZhao/TTP-2}.

\section{Conclusion}
We design two schedules to generate feasible solutions to TTP-2 with $n\equiv 0 \pmod 4$ and $n\equiv 2 \pmod 4$ separately, which can guarantee the traveling distance at most $(1+{\frac{3}{n}}-{\frac{10}{n(n-2)}})$ times of the optimal for the former case and $(1+{\frac{5}{n}}-{\frac{10}{n(n-2)}})$ for the latter case. Both improve the previous best approximation ratios.
For the sake of analysis, we adopt some randomized methods. The randomized algorithms can be derandomized efficiently with an extra running-time factor of $O(n^4)$. It is possible to improve it to $O(n^3)$ by using more detailed analysis (as we argued this for the case of even $n/2$ in \cite{DBLP:conf/cocoon/ZhaoX21}).
In addition to theoretical improvements, our algorithms are also very practical.
In the experiment, our schedules can beat the best-known solutions for all instances in the well-known benchmark of TTP with an average improvement of $2.86\%$ for even $n/2$ and an average improvement of $8.65\%$ for odd $n/2$. The improvements in both theory and practical are significant.

For further study, it should be interesting to extend the construction techniques and analyses in this paper to more TTP related problems, for example, TTP-$k$ with $k\geq 3$, single round-robin version of TTP-2 (STTP-2) \cite{imahori20211+}, linear distance relaxation of TTP-$k$ (LDTTP-$k$)~\cite{DBLP:journals/jair/HoshinoK12}, and so on.

\section*{Acknowledgments}
The work is supported by the National Natural Science Foundation of China, under grant 61972070.

\bibliographystyle{plain}
\bibliography{mybib}
\end{document}